\documentclass[11pt, oneside]{amsart}
\usepackage[foot]{amsaddr}
\usepackage{geometry}                
\usepackage{graphicx, slashed}
\usepackage{amssymb, amsthm, mathrsfs, bm, upgreek, mathabx, dsfont, marginnote}
\usepackage{epstopdf}
\usepackage{enumerate}
\usepackage{stmaryrd}
\usepackage{color}
\usepackage{hyperref}
\newcommand{\be}{\begin{equation*}}
\newcommand{\ee}{\end{equation*}}
\newcommand{\ben}[1]{\begin{equation}\label{#1}}
\newcommand{\een}{\end{equation}}
\newcommand{\bea}{\begin{eqnarray}}
\newcommand{\eea}{\end{eqnarray}}
\newcommand{\bean}{\begin{eqnarray*}}
\newcommand{\eean}{\end{eqnarray*}}

\newcommand{\R}{\mathbb{R}}
\newcommand{\T}{\mathbb{T}}

\newcommand{\C}{\mathbb{C}}
\newcommand{\p}{\partial}

\newcommand{\N}{\mathbb{N}}

\newcommand{\abs}[1]{\left|#1 \right|} 
\newcommand{\norm}[2]{\left \|#1 \right\|_{#2}}

\newcommand{\eq}[1]{(\ref{#1})}
\newtheorem{Theorem}{Theorem}
\newtheorem{Proposition}[Theorem]{Proposition}

\newtheorem{Lemma}[Theorem]{Lemma}

\newtheorem{Definition}[Theorem]{Definition}
\newtheorem*{T1}{Theorem~\ref{thm:EFTexp}}
\newtheorem*{T2}{Theorem~\ref{Thm:EFTApr}}

\title[EFT]{Effective field theory and classical equations of motion}

\author{Harvey S. Reall$^\dagger$}

\author{Claude M. Warnick$^{\dagger,\ddagger}$}

\address{$^\dagger$Department of Applied Mathematics and Theoretical Physics\\University of Cambridge \\Wilberforce Road, Cambridge CB3 0WA, United Kingdom}
\address{$^\ddagger$Department of Pure Mathematics and Mathematical Statistics\\University of Cambridge \\Wilberforce Road, Cambridge CB3 0WA, United Kingdom}

\email{hsr1000@cam.ac.uk,cmw50@cam.ac.uk}

\date{}                                           

\begin{document}

\begin{abstract}
Given a theory containing both heavy and light fields (the UV theory), a standard procedure is to integrate out the heavy field to obtain an effective field theory (EFT) for the light fields. Typically the EFT equations of motion consist of an expansion involving higher and higher derivatives of the fields, whose truncation at any finite order may not be well-posed. In this paper we address the question of how to make sense of the EFT equations of motion, and whether they provide a good approximation to the classical UV theory. We propose an approach to solving EFTs which leads to a well-posedness statement. For a particular choice of UV theory we rigorously derive the corresponding EFT and show that a large class of classical solutions to the UV theory are well approximated by EFT solutions. We also consider solutions of the UV theory which are not well approximated by EFT solutions and demonstrate that these are close, in an averaged sense, to solutions of a modified EFT.

\end{abstract}

\maketitle

\section{Introduction}

In this section we will give a non-technical introduction to our work. The next section will present an overview of our technical results.

Effective field theories (EFTs) play a central role in high energy physics. Given a ``UV" quantum field theory consisting of heavy and light fields, the idea is that one can ``integrate out" the heavy fields to obtain an EFT for the light fields (see e.g. \cite{Burgess:2003jk}). This EFT can be used to calculate scattering amplitudes of the light fields as an expansion in $E/(M\hbar)$ where $E$ is an energy scale characterizing the scattering and $M$ is a mass parameter characterizing the heavy fields, with dimensions of inverse length.\footnote{We will consider Lorentz invariant theories and use units with speed of light $c=1$.} The Lagrangian of an EFT consists of an infinite expansion of terms with increasing mass dimension (or increasing number of derivatives) built from the light fields. These terms are multiplied by appropriate inverse powers of $M$.  The idea is that only a finite number of terms in the Lagrangian are required to calculate scattering amplitudes to any desired order in $E/(M\hbar)$. 

Sometimes one wishes to study classical solutions of an EFT. If one truncates the EFT equations of motion to some finite order in $M^{-1}$ then the resulting equations involve higher than second derivatives of the light fields. There has been considerable discussion in the literature of how one should treat such equations, see e.g. \cite{Simon:1990ic,Flanagan:1996gw,Burgess:2014xh,Solomon:2017nlh,Allwright:2018rut}. For example, such equations typically admit ``runaway" solutions that blow up in time, and are usually regarded as unphysical. How do we eliminate such solutions? The nature of the initial value problem is unclear: with higher derivatives, do we need extra initial data? If so, is the initial value problem well-posed? Wouldn't additional initial data correspond to introducing new physical degrees of freedom? Should we only consider solutions which can be expanded in powers of $M^{-1}$? For a particularly clear discussion of these issues, see \cite{Flanagan:1996gw}. 

Our strategy for investigating these questions will be to determine the conditions under which classical solutions of the UV theory give rise to classical solutions of the EFT. Specifically, what conditions must the heavy and light fields of the UV theory satisfy in order that the light fields satisfy, at least approximately, the EFT equations of motion? We shall then view these conditions on the light fields as ``admissibility" criteria on solutions of the EFT equations of motion and use this to develop a notion of well-posedness for the initial value problem of the EFT equations. Of course this would be vacuous if our conditions on the fields of the UV theory were too strong, so the most substantial part of our paper will be devoted to demonstrating that the UV theory does indeed admit a ``suitably large" class of solutions satisfying these conditions. 

We shall focus throughout this paper on a particular UV theory discussed in \cite{Burgess:2003jk,Burgess:2014xh,Allwright:2018rut}. We expect that the lessons learned from studying this theory will extend to a much wider class of theories. This theory  is a complex scalar field with a $U(1)$-symmetric potential exhibiting spontaneous symmetry breaking. It can be rewritten in terms of two real fields: a heavy field $\rho$ with mass parameter $M$ and a massless field $\theta$ (a Goldstone boson). Integrating out $\rho$ consists of solving, formally, its equation of motion as an expansion in powers of $M^{-1}$ with coefficients that are polynomials in derivatives of $\theta$. Substituting this expansion into the UV equation of motion for $\theta$ gives the EFT equation of motion for $\theta$. This is an expansion in powers of $M^{-1}$, with higher powers of $M^{-1}$ involving more derivatives of $\theta$. 

Our first aim is to provide a mathematically rigorous derivation, in classical field theory, of these formal expansions in inverse powers of $M$ that are derived heuristically in the EFT literature. We view the fields as functions of $M$ as well as of the spacetime coordinates. We shall identify {\it uniform boundedness} as $M \rightarrow \infty$ of $(M\rho,\theta)$, and their derivatives, as the key assumption required to derive these expansions from the equations of motion of the UV theory.\footnote{
From a physics perspective, taking the limit $M \rightarrow \infty$ may appear strange since $M$ is usually regarded as a fixed parameter. However, as we shall explain, there is an equivalent viewpoint in which one fixes the mass of the heavy field to some value $\mu$ and considers fields which depend on the spacetime coordinates $x^\nu$ as $\mu x^\nu/M$. In this picture, one is considering a family of solutions (parameterised by $M$) varying over a length/time scale proportional to $M/\mu$. So $M \rightarrow \infty$ is a ``long wavelength" limit, and the EFT expansions hold for fields varying over scales that are increasingly large compared to the ``UV length scale" $\mu^{-1}$.}
 If this assumption is satisfied then $\rho$ admits an asymptotic expansion in powers of $M^{-1}$. Truncating this expansion at order $M^{-k}$ then implies that $\theta$ satisfies the EFT equation of motion up to terms of order $M^{-k-1}$.

Our second aim is to develop a notion of well-posedness for the EFT equation of motion, truncated to some order in $M^{-1}$. The idea is that, since uniform boundedness is required to derive the EFT expansion from the UV theory, then only uniformly bounded solutions of the EFT equation of motion should be admissible. Working on a fixed time interval $[0,T]$, we shall prove a uniqueness result: assume that $\theta_1$ and $\theta_2$ both satisfy the EFT equation of motion up to terms of order $M^{-k-1}$, that $\theta_2 - \theta_1$ and its first time derivative vanish initially, and that $\theta_1$ and $\theta_2$ have uniformly bounded derivatives as $M \rightarrow \infty$. Then we shall prove that $\theta_2 - \theta_1$ is of order $M^{-k-1}$. The boundedness assumption can be viewed as eliminating ``runaway" solutions. This is a uniqueness result for solutions of the EFT. We can also obtain an existence result for such solutions using the standard EFT method of constructing them as expansions in $M^{-1}$ up to some order. This solution depends continuously on the initial data. Taken together, these results can be viewed as a statement of well-posedness of the initial value problem, within the class of admissible fields, for the EFT equation of motion. Although our focus in this paper is on a particular EFT for $\theta$, we will prove these well-posedness results for a much more general class of EFTs for a single scalar field. 

In the results described so far, boundedness as $M \rightarrow \infty$ has always been assumed. Our third aim is to determine the circumstances (if any!) under which solutions of the UV theory actually satisfy this assumption, and thereby identify the circumstances under which solutions of the EFT approximate the behaviour of the light field in classical solutions of the UV theory. {\it A priori} it is not obvious that there are any circumstances under which solutions of an EFT approximate classical solutions of the underlying UV theory. This is because in quantum field theory we only expect EFT to be a good approximation when $E \ll M \hbar$, which means that there can be no heavy particles present. However, one expects that validity of the classical approximation for the heavy scalar field requires that a large number of heavy quanta are present, which implies $E \gg M \hbar$. So the expected regime of validity of the classical approximation in the UV theory does not overlap the expected regime of validity of EFT. This argument suggests that one cannot expect the EFT to describe the behaviour of (the light field in) classical solutions of the UV theory for some generic class of initial data. Instead, the best one can hope for is that the EFT approximates classical solutions of the UV theory arising from special initial data corresponding to the heavy field being ``in its ground state". Our boundedness asumptions can be regarded as a classical analogue of the statement that the heavy field is in its ground state. 

The question now is: if the boundedness assumptions, and hence the EFT expansions, are satisfied initially, then do they continue to be satisfied? We shall prove that this is indeed the case. More precisely: if the initial data for the light field is bounded as $M \rightarrow \infty$ and the initial data for the heavy field respects the EFT expansion truncated at order $M^{-k}$ then the resulting solution of the UV theory satisfies our boundedness assumptions and hence (by our previous results) the heavy field continues to respect the EFT expansion and the light field satisfies the EFT equation of motion to the expected order in $M^{-1}$. These results amount to a classical derivation of the EFT from the UV theory. 

In the above discussion we have assumed that we work on a time interval $[0,T]$ where $T$ is independent of $M$. Our fourth aim is to study what happens over long time scales, when $T$ scales with $M$: $T \propto M^\lambda$ with $\lambda>0$. The conventional approach of constructing EFT solutions as an expansion in powers of $M^{-1}$ can now fail because the coefficients in this expansion can exhibit secular growth. However, for $\lambda<2$ we shall show that one can construct EFT solutions as an expansion in $T/M^2$. We then revisit the question of whether such a solution approximates a solution of the UV theory. We shall argue that it does (for $\lambda<2$), assuming that the UV solution respects our boundedness assumptions and that the EFT equation of motion is satisfied to a sufficiently high order in $M^{-1}$ (the required order diverges as $\lambda \rightarrow 2$). However, we emphasise that we have not attempted to prove that solutions of the UV theory do indeed satisfy our boundedness assumptions on long time scales.

Our final aim is to investigate solutions of the UV theory arising from initial data for which the heavy field does {\it not} respect the EFT asymptotic expansion. For example: what happens if we simply require that this initial data is chosen such that the total energy is bounded as $M \rightarrow \infty$? (We now return to fixed $T$.) In this case, the heavy field will typically exhibit very rapid oscillations in time, which will be transmitted to the light field by interactions. This suggests that some kind of temporal averaging will be required to pass from the UV theory to an EFT for the light field. Averaging also appears very natural on physical grounds because low energy measurements will not be able to resolve very short time or length scales.
 
To warm up, we investigate averaging for some linear equations, specifically the Klein-Gordon equation with mass parameter $M$ coupled to a source, or the massless wave equation with a highly oscillatory source term. In both cases we shall show that, for a very large class of initial data, the field satisfies the EFT expansion in an averaged sense. 

In the nonlinear case, averaging is much more difficult. Nevertheless, for the $(\rho,\theta)$ system we shall prove that, for initial data of bounded energy density, there exists a modified EFT equation of motion, such that if $\tilde{\theta}$ satisfies this modified equation, with suitably modified initial data, then the average of $\theta-\tilde{\theta}$ over a finite spacetime region is subleading in $1/M$. The modifications to the EFT equation of motion and the initial data involve the initial data for the heavy field and vanish for the special initial data discussed above. We shall demonstrate that our result is sharp using numerical simulations. Specifically, we shall show numerically that solutions to the modified EFT equation approximate the UV solution in an averaged sense whereas solutions of the unmodifed EFT equation do not.

This concludes the non-technical introduction to this paper. The next section presents a detailed summary of our main technical achievements, with the proofs deferred to later sections.  

\subsection*{Acknowledgements} We are grateful to Istv\'an K\'ad\'ar and \'Aron Kov\'acs for comments on a draft manuscript. HSR is supported by STFC grant no. ST/T000694/1.

\section{Overview of results}

\label{sec:overview}

\subsection{Notation}

We consider a locally flat $(n+1)$-dimensional spacetime with coordinates $x^\mu = (x^0,x^i)$ ($i=1,\ldots,n$) and metric $g = - dx^0 dx^0 + \delta_{ij} dx^i dx^j$. Greek indices are raised and lowered with $g$. The spacetime manifold will be taken to be $\mathcal{M} = \R \times \T^n$ where $\T^n := \{ (x^1, \ldots, x^n) \in \R^n : x^i \sim x^i + L\}$ for fixed $L$. We will often consider a fixed time interval so for $T>0$ we let $\mathcal{M}_T = (0,T) \times \T^n$.  Our results will hold equally on the uncompactified space $\R \times \R^n$, provided we restrict attention to bounded regions, since we can make use of the finite speed of propagation property for wave equations to take $L$ sufficiently large that the periodic identification has no effect. We define the d'Alembertian as $\Box = g^{\mu\nu} \p_\mu \p_\nu$. 

For an open subset $\mathcal{R}$ of spacetime we will often consider functions of both $x^\mu \in \mathcal{R}$ and a mass parameter $M$. For notational convenience we will often not make the dependence on $M$ explicit. When we say $f=O(M^{-p})$ we mean that $ \sup_{\mathcal{R}} |f| \le CM^{-p}$ for some constant $C$ independent of $M$. Similarly for a function of $x^i \in \T^n$ and $M$ we say $f=O(M^{-p})$ if $ \sup_{\T^n} |f| \le CM^{-p}$. 


\subsection{The theory}

We will take as an example the UV theory discussed in \cite{Burgess:2003jk}: a complex scalar field $\phi$ with action\footnote{Throughout all integrals are with respect to the appropriate geometric measure induced by $g$.}
\be
 S = \int_\mathcal{M} \left[ -\frac{1}{2} \partial_\mu \phi \partial^\mu \bar{\phi} - V(|\phi|^2) \right]
\ee
where
\be
 V(|\phi |^2) = \frac{M^2}{8} \left( |\phi|^2 -1 \right)^2
\ee
and $M$ is a positive constant. This is a $U(1)$-symmetric ``Mexican hat" potential with a local maximum at $\phi=0$ and a global minimum at $|\phi|=1$. The Cauchy problem for this theory is:
\ben{WEq}
\begin{array}{c}
\displaystyle- \Box \phi + \frac{M^2}{2} \phi \left(\abs{\phi}^2 -1 \right) = 0 \\[.4cm]
\phi|_{x^0=0} = \phi_0, \qquad \p_0 \phi|_{x^0=0} = \phi_1
\end{array}
\een
We assume that $\phi_0, \phi_1$ are given smooth functions of $x^i$ and $M$. The resulting solution $\phi$ is a function of $x^\mu$ and $M$. 

The theory admits a conserved energy:
\[
E[\phi](t) = \int_{\{x^0 = t\}} \left[\frac{1}{2} \left( \abs{\p_0 \phi}^2 + \sum_{i=1}^n \abs{\p_i \phi}^2\right) +  V(|\phi |^2)\right],
\]
which, for smooth solutions, satisfies:
\[
\frac{d}{dt} E[\phi] = 0.
\]
If $n \leq 3$, then this energy gives sufficient control to establish that a smooth solution to \eq{WEq} exists for all time, for any choice of $\phi_0, \phi_1$. For general $n$ and arbitrary $\phi_0, \phi_1$, however, the best we can hope to show is that a smooth solution exists for a short time \cite{Ginibre:1985xy}. We are interested in understanding the behaviour as $M \rightarrow \infty$ of solutions to \eq{WEq} arising from initial data such that $E[\phi](0)$ is bounded as $M \rightarrow \infty$. Since $E[\phi]$ is conserved, this means that such solutions must satisfy $\abs{\phi} \to 1$ at least in some integrated sense.

In the mathematical literature, the limit $M \to \infty$ of this model has been studied when $\phi$ takes values in $\R^N$ as a means of constructing solutions to the $S^{N-1}$ wave map equation, also known as the $O(N)$ $\sigma$-model (see results in \cite{Shatah:1988},  based on the approach of \cite{Rubin:1957}). In the case $\phi \in \C$, $n=2$, the limit $M \to \infty$ can give rise to singular solutions containing vortices, whose dynamics has been studied in \cite{Jerrard:1999qt}. In \cite{Machihara:2002ay, Lin:2010fb} the system above is considered in various limits which have some similarities to our problem, but are nevertheless distinct.

We define real fields $(\rho,\theta)$ by\footnote{
Our definition of $\rho$, which differs from the conventional one \cite{Burgess:2003jk}, is chosen so that the range of $\rho$ is unrestricted. Strictly, $\theta$ should be thought of as a multi-valued function since it may be that on traversing a cycle of $\T^n$, $\theta$ increases by a multiple of $2 \pi$. Since all of our equations involve either $\p_\mu \theta$, or $\theta_1 - \theta_2$ for continuous $\theta_i$ which agree on  $\{x^0=0\}$, we can safely ignore this subtlety.}
\be
 \phi = e^{\rho + i \theta}
\ee
in terms of which the action becomes
\be
 S = \int_\mathcal{M}  \left[ -\frac{e^{2\rho}}{2} \left( \partial_\mu \rho \partial^\mu \rho + \partial_\mu \theta \partial^\mu \theta\right) - \frac{M^2}{8} \left( e^{2\rho}-1 \right)^2 \right]
\ee
and the resulting equations of motion are
\begin{align}
\displaystyle- \Box \rho + M^2 \rho &= \displaystyle \p_\mu \rho \p^\mu \rho - \p_\mu \theta  \p^\mu \theta- M^2 \rho^2 W(\rho), \qquad W(\rho) = \frac{(e^{2 \rho}-1 - 2 \rho) }{2\rho^2} \label{rhothetaeq1}\\[.1cm]
- \Box \theta  &= \ 2 \p^\mu \rho \p_\mu \theta. \label{rhothetaeq2}
\end{align}
Here we have explicitly separated the Klein--Gordon mass term for the $\rho$ equation, so that we see the $\rho$ field has mass parameter $M$ and the terms on the right-hand side of \eq{rhothetaeq1} are nonlinear. The function $W(\rho)$ extends smoothly on setting $W(0)=1$. The $\theta$ field is massless (a Goldstone boson).

Initial data for $\theta, \rho$ is given by
\[
\theta|_{\{x^0=0\}} = \theta_0, \quad \p_0 \theta|_{\{x^0=0\}} = \theta_1,\quad \rho|_{\{x^0=0\}} = \rho_0,\quad \p_0 \rho|_{\{x^0=0\}} = \rho_1,
\]
where this data is related to that for $\phi$ by:
\[
\phi_0 = e^{\rho_0 + i \theta_0}, \quad \phi_1 = (\rho_1 + i \theta_1)  e^{\rho_0 + i \theta_0}.
\]
The initial data depends smoothly on $x^i$ and $M$. As above, we view a solution $(\rho,\theta)$ as a function of $x^\mu$ and $M$. 

The conserved energy in terms of $(\rho,\theta)$ is
\be
 E[\rho,\theta] (t) = \int_{\{x^0 = t\}} \left\{\frac{e^{2\rho}}{2} \left[(\p_0 \rho)^2 + (\p_0 \theta)^2 + \sum_{i=1}^n \left( (\p_i \rho)^2 + (\p_i \theta)^2 \right) \right] +  \frac{M^2}{8} \left( e^{2\rho}-1 \right)^2 \right\}
\ee
If we take $M \rightarrow \infty$ for initial data such that $E$ remains bounded then we expect that the solution must satisfy $\rho \rightarrow 0$ in some appropriate sense.

We observe that the fact that global solutions to \eq{WEq} exist for $n = 3$ does not imply that (\ref{rhothetaeq1}), (\ref{rhothetaeq2}) admit global solutions. The map from $\phi$ to $(\rho, \theta)$ is only well defined provided $\phi$ does not vanish. In $3$-dimensions the conserved energy for $\phi$ does not give pointwise control of $\phi$, so even in the large $M$ limit a solution $\phi$ may vanish. At such a point solutions to  (\ref{rhothetaeq1}, \ref{rhothetaeq2}) will break down.

In this paper we shall be studying the behaviour of solutions as $M \rightarrow \infty$, with the initial data assumed to be suitably well-behaved in this limit. From a physics perspective, it might seem odd to consider a limit $M \rightarrow \infty$ since $M$ is a parameter of the theory, presumably fixed by physics. However, this limit can be reinterpreted in terms of a scaling of initial data with $M$ held fixed. To see this, let $(\rho(x),\theta(x))$ be a solution of \eq{rhothetaeq1},\eq{rhothetaeq2}. Now fix $\mu>0$ and define
\begin{equation}
\label{barredsystem}
 \bar{\rho}(x) = \rho\left(\frac{\mu x}{M}\right) \qquad \bar{\theta}(x)=\theta\left(\frac{\mu x}{M}\right)
\end{equation}
Then $(\bar{\rho},\bar{\theta})$ satisfy \eq{rhothetaeq1},\eq{rhothetaeq2} with $M$ replaced by $\mu$, and the periodicity of the spatial coordinates (if compactified) is now $x^i \sim x^i + M L /\mu$.  The initial data are
\bea
 \bar{\theta}(0,x^i) &=& \theta_0\left(\frac{\mu x^i}{M}\right), \qquad \partial_0 \bar{\theta}(0,x^i) = \frac{\mu}{M} \theta_1 \left(\frac{\mu x^i}{M}\right), \nonumber \\  \bar{\rho}(0,x^i) &=& \rho_0\left(\frac{\mu x^i}{M}\right), \qquad \partial_0 \bar{\rho}(0,x^i) = \frac{\mu}{M} \rho_1 \left(\frac{\mu x^i}{M}\right) \nonumber
\eea
If $(\rho,\theta)$ are bounded as $M \rightarrow \infty$ then the initial data for $(\bar{\rho},\bar{\theta})$, and the resulting solution, has first space and time derivatives of order $1/M$. The time interval on which $(\bar{\rho},\bar{\theta})$ are defined is $\bar{T} = MT/\mu$. Hence the limit $M \rightarrow \infty$ corresponds to a sequence of data/solutions varying over increasingly long wavelengths, and defined over increasingly long time intervals. The energy scales as $\bar{E} = (M/\mu)^{n-2} E$.

\subsection{The EFT expansion}

In the EFT literature, it is argued heuristically that \eqref{rhothetaeq1} implies that the heavy field $\rho$ admits an expansion in inverse powers of $M$, with coefficients that are polynomials in $\theta$ and its derivatives. This is then substituted into the action (or the equation of motion for $\theta$) to give the EFT for $\theta$. Our first result identifies a  sufficient condition to make this procedure rigorous. This is uniform \emph{boundedness} of $\theta$ and $M \rho$ as $M \to \infty$ in an appropriate sense. To motivate this, let us first suppose that we can bound $\theta$ and $M \rho$ \emph{and all their derivatives} uniformly for $M \geq M_0$.

We consider \eq{rhothetaeq1} and rearrange to obtain:
\ben{rho_eq_rearranged}
 \rho = -\frac{\partial_\mu \theta \partial^\mu \theta}{M^2}  +\frac{\Box \rho}{M^2}+ \frac{\partial_\mu \rho \partial^\mu \rho}{M^2}  -\rho^2 W(\rho)
\een
The uniform boundedness assumption implies that $\partial_\mu \theta = O(1)$, $\partial_\mu \rho = O(1/M)$ and $\partial_\mu \p_\nu \rho = O(1/M)$. Hence the first term on the RHS is $O(M^{-2})$, the second term is $O(M^{-3})$, and the third term is $O(M^{-4})$. For $M\ge M_0$ we know that $\rho$ is uniformly bounded and hence (since $W$ is smooth on $\mathbb{R}$) $W(\rho)$ is uniformly bounded, as are its derivatives $W'(\rho)$, $W''(\rho)$ etc. Therefore the final term is $O(M^{-2})$. Hence we have $\rho=O(M^{-2})$. Using this, our estimate of the final term can be improved to $O(M^{-4})$ and we see that the equation of motion implies $\rho = \mathcal{F}_2 M^{-2}  +O(M^{-3})$  with $\mathcal{F}_2 = - \p_\mu \theta \p^\mu \theta$. 

Next, take the derivative of \eqref{rho_eq_rearranged} and rearrange to obtain
\ben{drho}
 \partial_\nu \rho =- \frac{\partial_\nu(\partial_\mu \theta \partial^\mu \theta)}{M^2}  +\frac{\partial_\nu \Box \rho}{M^2}+ \frac{\partial_\nu(\partial_\mu \rho \partial^\mu \rho)}{M^2}  -\p_\nu(\rho^2 W(\rho))
\een
The second term on the RHS is $O(M^{-3})$ and the third is $O(M^{-4})$. The final term is the sum of $2\rho \partial_\nu \rho W(\rho) = O(M^{-3})$ and $\rho^2  W'(\rho) \partial_\nu \rho = O(M^{-5})$ (since $W'(\rho)$ is bounded). Hence $\p_\nu \rho = M^{-2} \partial_\nu \mathcal{F}_2 + O(M^{-3})$. Differentiating \eq{drho} again and estimating the terms we have similarly that $\p_\nu \p_\sigma \rho = M^{-2} \partial_\nu \p_\sigma \mathcal{F}_2 + O(M^{-3})$ and so on for higher derivatives. 

Returning to \eq{rho_eq_rearranged} we observe that by Taylor's theorem and the uniform boundedness of $M\rho$ we have $\abs{W(\rho) - 1} \leq C \rho \leq C'M^{-1}$, so that $\rho^2 W(\rho) = M^{-4}\mathcal{F}_2^2 + O(M^{-5})$. Combining this with our estimates for $\p_\mu \rho$ and $\p_\nu \p_\sigma \rho$ we can improve our approximation for $\rho$ to give $\rho =  \mathcal{F}_2 M^{-2}+ \mathcal{F}_4 M^{-4}  +O(M^{-5})$, with $\mathcal{F}_4 = \Box \mathcal{F}_2 - \mathcal{F}_2^2$. We can see that this process can be iterated to find a full asymptotic series for $\rho$.
\[
\rho \sim \frac{\mathcal{F}_2[\theta]}{M^2} + \frac{\mathcal{F}_4[\theta]}{M^4}+ \ldots 
\]
This series is also valid when differentiated term-by-term, so in view of the boundedness assumption on $\p_\mu \theta$, this implies that we have an expression for $\Box \theta$ as an asymptotic expansion
\ben{EFT_asympexp}
 \Box \theta \sim -2 \p_\mu \left( \frac{\mathcal{F}_2[\theta]}{M^2}+ \frac{\mathcal{F}_4[\theta]}{M^4} + \ldots  \right) \p^\mu \theta.
\een
We will refer to this expansion as the ``EFT equation of motion". It agrees with expansions derived by heuristic methods in the EFT literature.

In practice, the assumption that $\theta$ and $M\rho$ and all derivatives are bounded is a very strong one, and it is not clear that it is possible to exhibit non-trivial solutions of the equations for which this holds. If we are content to accept an expansion only to a given order, then it's sufficient to only require bounds on a finite number of \emph{time} derivatives. In particular, to obtain an asymptotic expansion with $O(M^{-(l+1)})$ error, we ask that $l$ time derivatives, and any number of spatial derivatives are bounded. We shall later give a condition on initial data such that these conditions hold, and we shall see that while bounding spatial derivatives is relatively straightforward if we assume the initial data is smooth, bounding time derivatives imposes non-trivial constraints on the initial data.

We will introduce some notation to make it a little simpler to keep track of our estimates. We say that $f = O_{\infty}(M^{-k})$ on $\mathcal{R}$ if for any $p\geq 0$ and any $i_1, \ldots, i_p \in \{1, \ldots, n\}$ there exists a constant $C_p$, independent of $M\geq M_0$ such that
\[
\sup_{\mathcal{R}} \abs{\p_{i_1} \ldots \p_{i_p} f} \leq \frac{C_p}{M^k}.
\]
Clearly if $f = O_{\infty}(M^{-k})$ and $g= O_{\infty}(M^{-k'})$ then we have $\p_i f = O_{\infty}(M^{-k})$, $fg= O_{\infty}(M^{-k-k'})$ and  $f+g = O_{\infty}(M^{-\min\{k, k'\}})$. This notation generalises in the obvious way to functions defined on subsets of $\T^n$.

\begin{Theorem}
\label{thm:EFTexp}
Let $(\rho,\theta)$ be a smooth solution of the equations of motion in a fixed open subset $\mathcal{R}$ of spacetime, satisfying 
\[
\theta,  \ \p_0\theta, \   M \rho, \   \p_0\rho = O_{\infty}(1).
\]
Fix a non-negative integer $l$. Then for all $0\leq k \leq l+2$ and $0\leq j \leq l$ we have
\ben{Bded1}
\p_0^k \theta , \    M\p_0^j \rho, \   \p_0^{l+1} \rho = O_{\infty}(1)
\een
if and only if
\ben{rho_exp_new1}
\p_0^j \rho= \p_0^j \left( \frac{\mathcal{F}_2[\theta]}{M^2}+ \frac{\mathcal{F}_4[\theta]}{M^4} + \ldots + \frac{\mathcal{F}_{2 \lfloor\frac{\ell-j}{2} \rfloor}[\theta]}{M^{2 \lfloor\frac{\ell-j}{2} \rfloor}} \right ) +O_{\infty}({M^{-\ell-1+j}})
\een
holds for $j\leq \ell+1\leq l +1$, where $\mathcal{F}_{2j}$ are polynomials in $\theta$ and its derivatives up to order $2j-1$ which can be computed iteratively.
\end{Theorem}
We have already established above that
\[
\mathcal{F}_2 = - \p_\mu \theta \p^\mu \theta,\qquad \mathcal{F}_4 = \Box \mathcal{F}_2 - \mathcal{F}_2^2.
\]
The expansion \eq{rho_exp_new1} in particular implies
\be
\rho = \frac{\mathcal{F}_2[\theta]}{M^2}+ \frac{\mathcal{F}_4[\theta]}{M^4} + \ldots + \frac{\mathcal{F}_{2 \lfloor\frac{l}{2} \rfloor}[\theta]}{M^{2 \lfloor\frac{l}{2} \rfloor}} + O_\infty(M^{-l-1})
\ee
and
\ben{EFT_exp}
\Box \theta + 2 \p_\mu \left( \frac{\mathcal{F}_2[\theta]}{M^2}+ \frac{\mathcal{F}_4[\theta]}{M^4} + \ldots +  \frac{\mathcal{F}_{2 \lfloor\frac{l}{2} \rfloor}[\theta]}{M^{2 \lfloor\frac{l}{2} \rfloor}}  \right) \p^\mu \theta = O_\infty(M^{-l}).
\een

To derive these results we need the assumption \eqref{Bded1}, which excludes a large class of solutions. For example, consider the simpler situation of the {\it free} Klein Gordon equation $-\Box \rho + M^2 \rho = 0$. Then $\rho=M^{-1} \sin(Mx^0)$ is a solution for which $M\rho$ is uniformly bounded but $M\partial_0 \rho$ is not. However, $\rho = e^{-M/M_0}\sin(Mx^0)$ is a solution for which all derivatives of $M\rho$ are uniformly bounded. Roughly speaking, uniform boundedness captures the notion that the ``heavy" degree of freedom ``tends to its ground state at an appropriate rate" as $M \rightarrow \infty$. We see that the correct heavy degree of freedom does not correspond to $\rho$ to all orders in $M$. Our separation of the system into $\rho$ and $\theta$ was somewhat arbitrary. The result above suggests that to order $M^{-2}$, the corrected variable $\overline{\rho}_2:=\rho - M^{-2} \mathcal{F}_2[\theta]$ is a better representative of the ``true" heavy degree of freedom in the system than $\rho$, and so on at higher orders. 

\subsection{EFT solutions}

Next we digress slightly to discuss what it means to solve an equation of the form \eqref{EFT_asympexp}. If we truncate this equation, retaining terms up to $\mathcal{F}_{2n}$ then (for $n>1$) we obtain an equation that is higher than second order in derivatives of $\theta$. This raises well-known problems: a higher order equation requires extra initial data and is unlikely to admit a well-posed initial value problem. In the ODE context, it is well-known that higher order equations can give rise to spurious ``runaway" solutions that blow up in time. We will show that the assumption of uniform boundedness as $M \rightarrow \infty$ provides a natural way to address these issues. 

We will consider a generalization of the EFT equation of motion \eqref{EFT_asympexp}:
\be
 -\Box \theta + m^2 \theta \sim \sum_{k=1}^\infty \frac{S_k(\theta,\partial \theta,\ldots)}{M^k} 
 \ee
where $m$ is a fixed ($M$-independent) constant, possibly zero, and each $S_k$ is a polynomial in $\theta$ and finitely many of its derivatives, with coefficients independent of $M$. For this equation we make the following definition: 

\begin{Definition}\label{def:EFT} Let $\mathcal{R}$ be an open region of spacetime. Let $\theta$ be a smooth function of $x^\mu \in \mathcal{R}$ and $M$. We say that $\theta$ is an EFT solution to order $l \in \{0,1,2,\ldots\}$ on $\mathcal{R}$ if $\theta$ and all of its derivatives (w.r.t. $x^\mu$) are uniformly bounded as $M \rightarrow \infty$ and $\theta$ satisfies 
\ben{genEFT}
 -\Box \theta + m^2 \theta = \sum_{k=1}^l \frac{S_k(\theta,\partial \theta,\ldots)}{M^k} + \frac{R_l}{M^{l+1}}
 \een
where $R_l$ and all of its derivatives are uniformly bounded as $M \rightarrow \infty$.
\label{def:EFT_sol}
\end{Definition}

Note that we are now requiring boundedness of arbitrarily many time derivatives of $R_l$ so this definition makes stronger assumptions about $\theta$ than Theorem \ref{thm:EFTexp}. We will comment further on this below. We will prove:\footnote{The assumption that the initial data are defined on $\mathbb{T}^{n}$ can be replaced by the assumption that the data are defined on $\mathbb{R}^n$ and have compact support.}

\begin{Theorem} Let $\theta_0,\theta_1$ be smooth functions of $x^i \in \mathbb{T}^{n}$ and $M$ such that $\theta_0,\theta_1$ and their derivatives (w.r.t. $x^i$) are uniformly bounded as $M \rightarrow \infty$. Then, for any ($M$-independent) $T>0$, and any $l$,
there exists an EFT solution $\theta$ of order $l$ on $\mathcal{M}_T$ with $\theta|_{x^0=0} = \theta_0$ and $\partial_0 \theta|_{x^0=0} = \theta_1$. Furthermore if $\tilde{\theta}$ is any other such solution then $\tilde{\theta}-\theta$ and all of its derivatives are $O(M^{-l-1})$. 
\label{thm:EFT_sol}
\end{Theorem}

The ``existence" part of the proof of this theorem uses the standard EFT approach of constructing a solution as an expansion in powers of $M^{-1}$. Notice that if $\theta$  satisfies the conditions of the theorem, then so does $\theta +M^{-(l+1)} \chi$ for any $\chi$ which is uniformly bounded together with all its derivatives as $M \to \infty$, and which vanishes to first order on $\{x^0=0\}$. Thus we cannot expect a stronger statement of uniqueness than we obtain. The proof of uniqueness exploits heavily the uniform boundedness assumption, which can be regarded as eliminating ``spurious" solutions. We emphasize that we mean boundedness as $M \rightarrow \infty$, not boundedness as $T\rightarrow \infty$.\footnote{
Although $M \rightarrow \infty$ does imply $T \rightarrow \infty$ for the $(\bar{\rho},\bar{\theta})$ system \eqref{barredsystem} with fixed mass parameter $\mu$.} Crucially, this result gives us existence and uniqueness \emph{without} requiring additional initial data, even though the truncated EFT equation may involve more than two time derivatives. 

We will also describe below how we can give a topology to the space of EFT solutions and the Cauchy data such that the solution depends continuously on the initial data. Together with Theorem \ref{thm:EFT_sol}, this constitutes a proof of well-posedness of the initial value problem, for EFT solutions.

\subsection{Approximation of UV solutions by EFT solutions}

In Theorem \ref{thm:EFTexp}, when deriving the EFT equation of motion for our particular UV problem above, we worked only to finite order in time derivatives \eq{EFT_exp}. However, when introducing the notion of EFT solution in Definition \ref{def:EFT}, we required boundedness of {\it all} time derivatives of $\theta$ and $R_l$. We could have chosen to only require boundedness for a finite number of time derivatives of $\theta$ and $R_l$. This would have the advantage that if $(\rho,\theta)$ is a solution of the UV theory satisfying the conditions of Theorem \ref{thm:EFTexp} then $\theta$ would define an EFT solution to an appropriate order. However, the disadvantage is that a more complicated definition of an EFT solution would make establishing a uniqueness statement analogous to Theorem \ref{thm:EFT_sol} considerably harder, and more dependent on the detailed form of the terms $S_k$ in \eqref{genEFT}.

We have chosen to use a straightforward definition for an EFT solution, as a consequence of which the light field in a UV solution satisfying Theorem \ref{thm:EFTexp} is not itself an EFT solution. As a result of this, we cannot directly appeal to Theorem \ref{thm:EFT_sol} to establish that such a solution of the UV theory must be close to an EFT solution. Instead we will directly prove the following result concerning the approximation of solutions to the UV theory by EFT solutions.
\begin{Theorem} \label{Thm:EFTApr}
Suppose that $(\rho, \theta)$ satisfy the conclusions of Theorem  \ref{thm:EFTexp} on $\mathcal{M}_T$ for some $l$, and suppose that $\tilde{\theta}$ is an EFT solution of \eq{EFT_asympexp} to order $l$ such that
\[
\theta|_{x^0=0}= \tilde{\theta}|_{x^0=0}, \qquad \p_0\theta|_{x^0=0}= \p_0\tilde{\theta}|_{x^0=0}.
\]
Then for any  $j\leq l$ we have
\[
\p_0^{j}  \theta =  \p_0^j \tilde{\theta}  + O_{\infty}(M^{-l-1+j}).
\]
In particular, $\theta-\tilde{\theta} = O_\infty (M^{-l-1})$.
\end{Theorem}

We note that although we establish existence of EFT solutions using a perturbation series approach in Theorem \ref{thm:EFT_sol}, Theorem \ref{Thm:EFTApr} is agnostic about how the EFT solution is constructed. Whether we solve the truncated EFT equation while choosing the additional initial data to impose the boundedness condition, or we use a reduction of order technique and solve a modified EFT equation which is second order in time, or take some other approach is immaterial. Any valid approach to construct an EFT solution satisfying Definition \ref{def:EFT} will result in an approximation to the UV solution to the appropriate order.

\subsection{The nonlinear truncation}

We now return to the specific EFT of equation \eqref{EFT_asympexp}. If we truncate this equation, retaining only the leading order EFT corrections then the resulting equation can be written
\ben{nonlinear_wave}
 G^{\mu\nu} \partial_\mu \partial_\nu \theta = 0 \qquad \qquad G^{\mu\nu} \equiv g^{\mu\nu} - \frac{4}{M^2} \partial^\mu \theta \partial^\nu \theta
\een
note that this does not involve higher than second derivatives of $\theta$: it is a non-linear wave equation, which is hyperbolic for sufficiently large $M$, assuming that $\partial_\mu \theta$ is bounded as $M \rightarrow \infty$. Thus, for initial data as considered above, this equation admits an initial value problem that is locally well-posed in the conventional sense. Any smooth solution of \eqref{nonlinear_wave} whose derivatives are all uniformly bounded as $M \rightarrow \infty$ is an EFT solution of order $3$ (because $S_3$ vanishes for this particular EFT). It is natural to ask what we can learn from the fact that we can obtain an EFT solution from a nonlinear wave equation. In particular, for \eqref{nonlinear_wave}, we can define an ``effective metric" $(G^{-1})_{\mu\nu}$, the inverse of $G^{\mu\nu}$. This effective metric determines causal properties of \eqref{nonlinear_wave}. The null cone of this effective metric lies everywhere inside, or on, the null cone of the Minkowski metric $g_{\mu\nu}$ (see e.g. \cite{Adams:2006sv}) so \eqref{nonlinear_wave} describes disturbances propagating below, or at, the speed of light. One might wonder whether we can exploit Theorem \ref{thm:EFT_sol} to deduce that $(G^{-1})_{\mu\nu}$ provides a better description of causality than $g_{\mu\nu}$ for a general EFT solution. However, we shall argue below that, over a fixed ($M$-independent) time interval, a general EFT solution cannot distinguish between the domains of dependence defined by these two different metrics.

\subsection{Long timescales}

 In the discussion so far, we assume that the time interval $[0, T)$ on which we work is fixed independently of $M$. In some circumstances one may be interested in timescales which grow with $M$, say $T \propto M^\lambda$ for some $\lambda>0$. In general, we have no reason to expect that a particular choice of initial data will give rise to a UV solution with this lifetime, and even if such a solution does exist, we have no reason to expect that it will continue to satisfy the boundedness conditions of Theorem \ref{thm:EFTexp}. However, it may be that in some circumstances\footnote{For example, one might hope that with a decompactified domain ($\R^n$ replacing $\T^n$), one could make use of the decay mechanism provided by dispersion to establish long time existence and boundedness, at least for data which is small in some appropriate sense.}there does exist a long-lived solution satisfying the conditions of Theorem \ref{thm:EFTexp}. Our approach above can be adapted to this situation, however it now becomes necessary to keep track of $T$ in the various estimates. 

For simplicity we consider only the first correction in the EFT expansion. In the context of the particular UV theory described by \eq{rhothetaeq1}, \eq{rhothetaeq2}, we will show that an EFT solution to order $3$ does exist for $\lambda <2$. However,  in order to construct such a solution through iterated solution of linear equations the effective expansion parameter becomes $T/M^2$ which means that we have to perform more iterations than in the case for which $T$ is independent of $M$. In particular, if $\lambda$ is close to $2$ then many iterations are required. By contrast, if we work with \eqref{nonlinear_wave}, then it suffices to solve this single nonlinear equation to obtain an EFT solution to order $3$, assuming that the solution satisfies our boundedness conditions (which we have not attempted to prove). So \eqref{nonlinear_wave} seems better suited to constructing EFT solutions over such long time intervals than the iterative approach.

For $\lambda <4/3$ we shall establish the following result showing that EFT solutions do approximate the UV solution, but that the accuracy of the EFT approximation is lessened over long timescales.
\begin{Lemma}\label{LongTLemma}
Let $0<T \leq C M^\lambda$ for some $0\leq \lambda <  \frac{4}{3}$. Suppose $(\rho, \theta)$ solve the UV equations of motion on the interval $[0,T)$, and satisfy the conditions of Theorem \ref{thm:EFTexp} to order $l=3$. Let $\tilde{\theta}$ be an EFT solution of \eq{EFT_asympexp} to order $3$ such that
\[
\theta|_{x^0=0}= \tilde{\theta}|_{x^0=0}, \qquad \p_0\theta|_{x^0=0}= \p_0\tilde{\theta}|_{x^0=0}.
\]
Then we have
\[
\theta= \tilde{\theta} + O_\infty \left( \frac{T^3}{M^4} \right).
\]
\end{Lemma}

We also claim (without proof) that this result can be improved as follows. If the UV solution satisfies the conditions of Theorem \ref{thm:EFTexp} to order $l$ and $\tilde{\theta}$ an EFT solution to order $l$ then $\tilde{\theta}-\theta$ is $O_\infty(T^{\lfloor \frac{l}{2}\rfloor+2}/M^{l+1})$. This implies that, for any $\lambda<2$, if $l$ is large enough then the EFT solution approximates the UV solution.

\subsection{Behaviour of solutions of the UV theory}

In order to obtain the EFT equation of motion we needed to make some strong assumptions in Theorem \ref{thm:EFTexp}. 
We now discuss the question of whether there exists a suitably large class of solutions of the UV theory that satisfies these assumptions. So we need to understand the behaviour of solutions of equations \eq{rhothetaeq1}, \eq{rhothetaeq2}  as $M \rightarrow \infty$. 

The equations \eq{rhothetaeq1}, \eq{rhothetaeq2} are nonlinear, and so for general dimension and arbitrary initial data we do not expect a solution to exist for all time. Standard local well-posedness results ensure that for any fixed $M$ there exists a time $T>0$ for which a smooth solution exists in the time interval $(0, T)$. However we need to exclude the possibility that $T \rightarrow 0$ as $M \rightarrow \infty$. We need a lower bound on the time of existence that is independent of $M$, at least for $M$ sufficiently large. 
\begin{Proposition}\label{Prop1}
Suppose that for each $M \geq M_0$ the initial data functions $M\rho_0, \rho_1, \theta_0, \theta_1$ are smooth and $O_{\infty}(1)$. Then for $M_0$ sufficiently large, a smooth solution to \eq{rhothetaeq1}, \eq{rhothetaeq2} exists on some time interval $[0, T)$, where $T>0$ may be taken independent of $M$. Furthermore, we have:
\[
\theta, \ \p_0\theta,\   \p_0^2\theta,\   M \rho, \   \p_0\rho  = O_{\infty}(1) 
\]
in $\mathcal{M}_T$, where the implicit constants depend on the initial data.
\end{Proposition}
We note that the assumption that initial data are smooth is stronger than necessary -- in practice we can work at a finite order of differentiability for the data and the solution, but for clarity it is simpler to assume spatial smoothness.

At this stage we can compare the bounds above with those required by Theorem \ref{thm:EFTexp}.  Proposition \ref{Prop1} gives boundedness of $M\rho$ and $\theta$, provided we take only {\it spatial} derivatives. It does not imply that $M\p_0 \rho$ is bounded, merely that $\p_0 \rho$ is, and it makes no claim regarding $\p_0^2 \rho$. In fact, \eq{rhothetaeq1} gives (assuming the bounds of Proposition \ref{Prop1}):
\ben{oscil}
\p_0^2 \rho + M^2 \rho = O_\infty(1).
\een
In general, we should expect both terms on the left-hand side to be $O_\infty(M)$, so the assumption of uniform boundedness of $M \partial_0^2 \rho$ cannot be valid in general. Equation \eq{oscil} suggests that at top order, at each spatial point $x^i$, $\rho$ satisfies a simple harmonic oscillator equation with frequency $M$. We can expect roughly that solutions arising from initial data as in Proposition \ref{Prop1} will behave as $\rho \sim M^{-1} \kappa(x^i) \sin (Mx^0+\varphi(x^i))$, with $\kappa, \varphi$ determined from the initial data. Therefore we expect that such solutions will have $\p_0^{k}\rho \sim M^{k-1}$.

We should not be surprised that solutions arising from generic data as in Proposition \ref{Prop1} do not have uniformly bounded time derivatives. Since $\{x^0=0\}$ is a non-characteristic surface, for $l\geq 2$ we can find expressions for $\rho_l:=\p_0^l \rho|_{\{x^0=0\}}$ and $\theta_l:= \p_0^l \theta|_{\{x^0=0\}}$ in terms of $\rho_0, \rho_1, \theta_0, \theta_1$ and their derivatives by using equations \eq{rhothetaeq1}, \eq{rhothetaeq2} to replace $\p_0^2$ derivatives. It is clear that  if $M\p_0^l \rho$, $\p_0^l \theta$ in $\mathcal{M}_T$ are uniformly bounded in $M$, then  $M\rho_l$ and $\theta_l$ must be bounded on $\T^n$, so we expect some constraints on initial data. In fact, we can find a necessary and sufficient condition for uniform boundedness of $M\p_0^l \rho$, $\p_0^l \theta$ in $\mathcal{M}_T$ in terms of initial data.

\begin{Proposition}\label{Prop2}
Suppose that for each $M>M_0$ we choose initial data such that $ \theta_0, \theta_{1}$ and $M\rho_j, \rho_{j+1}$ for $0\leq j \leq l$ are smooth and $O_\infty(1)$. Then if $T$ is as in Proposition \ref{Prop1}, for $0\leq k \leq l+2$ and $0\leq j \leq l$ we have
\ben{Bded2}
\p_0^k \theta, \  M\p_0^j \rho, \ \p_0^{l+1} \rho = O_{\infty}(1)
\een
on $\mathcal{M}_T$, with the implicit constants depending on the initial data.
\end{Proposition}

Systems of symmetric hyperbolic PDE with singular limits similar to that appearing in (\ref{rhothetaeq1}, \ref{rhothetaeq2}) have been considered in the mathematical literature, in particular with reference to incompressible limits of compressible fluids \cite{Kr, BrKr, KlMj,Sch}. Versions of some of our results above could be derived from those present in these papers, although a naive translation of our equations to symmetric hyperbolic form would result in working at higher order in derivatives. More concretely, our Proposition \ref{Prop2} could be obtained by casting (\ref{rhothetaeq1}, \ref{rhothetaeq2}) as a symmetric hyperbolic system and appealing to the results of \cite{Kr, BrKr,KlMj}. In particular, this result is an example of the ``bounded derivative principle'' elucidated by Kreiss \cite{Kr}.

The hypotheses of this proposition can be re-expressed as a condition on $\rho_0, \rho_1$, which essentially requires that the EFT expansion for $\rho$ is satisfied to a certain order at the level of the initial data.
\begin{Lemma} \label{ICLem}
Suppose that $\theta_0, \theta_1$ are given smooth functions, which are $O_\infty(1)$. Then for $0\leq j \leq l$ the functions $M\rho_j, \rho_{j+1}$ are smooth and have derivatives which are uniformly bounded in $M$ if and only if
\begin{align*}
\rho_0 &= \frac{P_2[\theta_0, \theta_1]}{M^2} + \ldots + \frac{P_{2 \lfloor\frac{l}{2} \rfloor}[\theta_0, \theta_1]}{M^{2 \lfloor\frac{l}{2} \rfloor}} + O_{\infty}(M^{-l-1}) \\
\rho_1 &= \frac{\tilde{P}_2[\theta_0, \theta_1]}{M^2} + \ldots + \frac{\tilde{P}_{2 \lfloor\frac{l-1}{2} \rfloor} [\theta_0, \theta_1]}{M^{2 \lfloor\frac{l-1}{2} \rfloor}} + O_{\infty}(M^{-l}) 
\end{align*}
Here $P_j[\theta_0, \theta_1], \tilde{P}_j[\theta_0, \theta_1]$ are explicit polynomials of $\theta_0, \theta_1$ and their derivatives which can be determined iteratively. The first terms in the expansion are:
\[
P_2[\theta_0, \theta_1] = \theta_1^2-\p_i \theta_0 \p_i \theta_0, \qquad \tilde{P}_2[\theta_0, \theta_1] = 2 \theta_1 \p_i \p_i \theta_0 - 2\p_i \theta_0 \p_i \theta_1.
\]
\end{Lemma}

The above expressions for $\rho_0$ and $\rho_1$ are the same as those obtained from the EFT expansion \eqref{rho_exp_new1} and using a formal ``reduction of order" approach in which the equations of motion are used to eliminate any second and higher time-derivatives of $\theta$. 

Combining these two results, we see that for any smooth and bounded $\theta_0, \theta_1$ we can find a solution of the UV equations for which the conditions of Theorem \ref{thm:EFTexp} are satisfied. In particular, this means that \emph{every} EFT solution arises as an approximation to a full solution to the UV equations. 

The converse statement -- that every solution to the UV equations is well described by an EFT solution -- is clearly \emph{not} true. If we do not insist that $\rho_0$, $\rho_1$ are given by the expressions in Lemma \ref{ICLem} then we do not expect that $\theta$ will be well approximated by an EFT solution. The issue, as described below Proposition \ref{Prop1}, is that we expect $\rho$ to be highly oscillatory, with frequency $M$ and nonlinearity will transmit these oscillations to $\theta$. In the language of \cite{Sch}, this is a `fast' limit, rather than the `slow' limits that result when the initial data is prepared according to Proposition \ref{Prop2}. For most practical purposes, these rapid oscillations will be very difficult to observe, which suggests that $\theta$ might still be approximated by an EFT solution, but only in an  {\it averaged} sense. We shall show that this is indeed the case, at least to leading non-trivial order, provided we modify the EFT equation of motion and the initial data in a suitable way. 

Assume that the initial data is as in Proposition \ref{Prop1}. Let $\varepsilon = \frac{1}{2}\left(\rho_1^2 + M^2 \rho_0^2 \right)$ be the initial energy density of the $\rho$ field. We now define a modified EFT solution by solving iteratively the following linear wave equations:
\ben{LinApr1}
\begin{array}{crc}
\displaystyle- \Box \tilde{\theta}^0 = 0,&\qquad& \displaystyle- \Box \tilde{\theta}^1 = 0, \\[.4cm]
\tilde{\theta}^0 |_{x^0=0} = \theta_0, &\qquad & \tilde{\theta}^1 |_{x^0=0} = 0, \\[.4cm]
\p_0 \tilde{\theta}^0 |_{x^0=0} = \theta_1& \qquad &\p_0 \tilde{\theta}^1 |_{x^0=0} = 2M \rho_0 \theta_1.
\end{array}
\een
and
\ben{LinApr2}
\begin{array}{c}
 \displaystyle- \Box \tilde{\theta}^{2} =  -6 \p_\mu\left(\varepsilon \p^\mu \tilde{\theta}^0 \right) -2\p^\mu\left( \p_\nu \tilde{\theta}^0 \p^\nu \tilde{\theta}^0 \p_\mu \tilde{\theta}^0 \right), \\[.4cm]
 \tilde{\theta}^{2} |_{x^0=0} = -2\rho_1 \theta_1 , \\[.4cm]
\p_0 \tilde{\theta}^{2} |_{x^0=0} = 2 M^2 \rho_0^2  \theta_1+ 2( \p_i \theta_0 \p^i \theta_0- \theta_1^2) \theta_1 +2\p_i  (\rho_1 \p_i \theta_0)+6 \varepsilon \theta_1.
\end{array}
\een
We shall show that $\tilde{\theta}^0 + \tilde{\theta}^1/M + \tilde{\theta}^2/M^2$ approximates $\theta$ to order $M^{-2}$ when averaged against an arbitrary compactly supported weight function $\eta$:

\begin{Theorem}\label{Prop3}
Let the initial data and $T>0$ be as in Proposition \ref{Prop1}. Then for any  weight function $\eta \in C^\infty_c(\mathcal{M}_T)$, we have that:
\[
\int_{\mathcal{M}_T} \left(\theta - \tilde{\theta}^0- \frac{1}{M}\tilde{\theta}^1- \frac{1}{M^2}\tilde{\theta}^{2}  \right) \eta = O(M^{-3}).
\]
\end{Theorem}
Rather than constructing our approximation by iterated linear equations we can, if we prefer, solve a single nonlinear equation. Note that if $M$ is sufficiently large, there will exist a solution  $\hat{\theta}^2:\mathcal{M}_T\to \R$  to the nonlinear wave equation:
\ben{NLinApr}
\begin{array}{c}
 \displaystyle- \Box \hat{\theta}^2 =  -\frac{6}{M^2} \p_\mu\left(\varepsilon \p^\mu \hat{\theta}^2 \right) -\frac{2}{M^2}\p^\mu\left( \p_\nu \hat{\theta}^2 \p^\nu \hat{\theta}^2 \p_\mu \hat{\theta}^2 \right), \\[.4cm]
 \hat{\theta}^2 |_{x^0=0} = \theta_0 - \frac{2}{M^2} \rho_1 \theta_1, \\[.4cm]
\p_0 \hat{\theta}^2 |_{x^0=0} = \theta_1\left( 1+ 2 \rho_0  +2 \rho_0^2+ \frac{2}{M^2} \left( \p_i \theta_0 \p^i \theta_0- \theta_1^2\right)  +\frac{6}{M^2} \varepsilon \right)+\frac{2}{M^2}\p_i  (\rho_1 \p_i \theta_0).
\end{array}
\een
Then one can show that
\[
\int_{\mathcal{M}_T} \left(\hat{\theta}^2 - \tilde{\theta}^0- \frac{1}{M}\tilde{\theta}^1- \frac{1}{M^2}\tilde{\theta}^{2}  \right) \eta = O(M^{-3}),
\]
so that $\hat{\theta}^2$ also approximates $\theta$ to $O(M^{-3})$ in the sense of averages. 

Comparing the above equation with the defining equation of an EFT solution of order $2$, we see that they differ via the term involving $\varepsilon$. The initial data is also modified via the terms involving $\rho_0$ and $\rho_1$. However, these differences become subleading when $\rho_1 = O_\infty (M^{-2})$ and $\rho_0 = \frac{1}{M^2} \left( \theta_1^2-\p_i \theta_0 \p^i \theta_0\right)+O_\infty(M^{-3})$, in particular they are subleading under the conditions of Lemma \ref{ICLem} (with $l=2$) when the standard EFT approach without averaging should be valid. Thus Theorem \ref{Prop3} is consistent with the EFT results above. 


While in this setting we have only constructed an approximation to $\theta$ to $O(M^{-2})$, we see in principle no obstacle to continuing to higher orders in the asymptotic expansion, however even to find the terms of $O(M^{-2})$ involves some complicated manipulations, and this is compounded as we move to higher order.

\subsection*{Layout of the paper} 
In \S \ref{BdedEFT} we establish the results above which connect the assumptions of uniform boundedness to the validity of the EFT expansion, and to the existence and uniqueness of an EFT solution. In \S\ref{IVP} we establish the existence of a large class of solutions which satisfy the necessary boundedness assumptions for the EFT expansion to be valid, and we show that under these conditions the EFT solution approximates the UV solution to the expected degree of accuracy. In \S\ref{LTS} we consider the existence and accuracy of EFTs over longer timescales. In \S\ref{Average} we discuss the process of establishing estimates for the averaged solution and show that in general a modified EFT equation is necessary to describe the UV solution, if we do not assume our initial data for $\rho$ is carefully prepared. Finally in \S\ref{Numerics} we present some numerics that confirm our previous analytical results.

\section{Boundedness and the EFT expansion} \label{BdedEFT}

\subsection{Proof of Theorem \ref{thm:EFTexp}}

We recall
\begin{T1}
Let $(\rho,\theta)$ be a smooth solution of the equations of motion in a fixed open subset $\mathcal{R}$ of spacetime, satisfying 
\[
\theta,  \ \p_0\theta, \   M \rho, \   \p_0\rho = O_{\infty}(1).
\]
Fix a non-negative integer $l$. Then for all $0\leq k \leq l+2$ and $0\leq j \leq l$ we have
\ben{Bded}
\p_0^k \theta , \    M\p_0^j \rho, \   \p_0^{l+1} \rho = O_{\infty}(1)
\een
if and only if
\ben{rho_exp_new}
\p_0^j \rho= \p_0^j \left( \frac{\mathcal{F}_2[\theta]}{M^2}+ \frac{\mathcal{F}_4[\theta]}{M^4} + \ldots + \frac{\mathcal{F}_{2 \lfloor\frac{\ell-j}{2} \rfloor}[\theta]}{M^{2 \lfloor\frac{\ell-j}{2} \rfloor}} \right ) +O_{\infty}({M^{-\ell-1+j}})
\een
holds for $j\leq \ell+1\leq l +1$, where $\mathcal{F}_{2j}$ are polynomials in $\theta$ and its derivatives up to order $2j-1$ which can be computed iteratively.
\end{T1}
\begin{proof}
We first show that \eq{Bded} implies \eq{rho_exp_new}. The proof is a double induction, first on $l \in \N_0$ and then on $j\leq l$. To start the induction on $l$ we note that the case $l=0$ is trivially true. Now suppose the result holds for $l=L\geq 0$, and that $\rho, \theta$ satisfy \eq{Bded}  for $l=L+1$. We need to establish that \eq{rho_exp_new} holds for $\ell = L+1$, since the case $\ell \leq L$ follows from the induction assumption.

The assumption that \eq{Bded} holds for $l=L+1$ immediately implies that \eq{rho_exp_new} holds for $\ell=L+1$ and $j=L+1, L+2$. We now differentiate \eq{rho_eq_rearranged} $J$ times with respect to $x^0$ to obtain:
\ben{dJrho}
\p_0^{J} \rho = -\frac{\p_0^{J}(\partial_\mu \theta \partial^\mu \theta)}{M^2}  +\frac{\Box \p_0^{J} \rho}{M^2}+ \frac{\p_0^{J}(\partial_\mu \rho \partial^\mu \rho)}{M^2}  -\p_0^{J}(\rho^2 W(\rho))
\een
For $J=L$, we see that the first, second and third terms are $O_{\infty}(M^{-2})$ from the bounds we have assumed on $\theta$ and $\rho$ ($\theta$ is differentiated at most $L+1$ times and $\rho$ at most $L+2$ times with respect to $x^0$). Expanding out the final product by Leibniz rule, and using the fact that $\p_0^j \rho = O_{\infty}(M^{-1})$ for $j \leq L+1$ we see that $\p_0^{L} \rho = O_{\infty}(M^{-2})$, which establishes  \eq{rho_exp_new} for $\ell=L+1$ and $j=L$.

Now, let $J < L$ and suppose that we have established \eq{rho_exp_new} for $\ell=L+1$ and $J+1\leq j \leq L+1$. We now wish to use \eq{rho_exp_new} to replace $\rho$ on the right-hand side of \eq{dJrho}. In the second term on the right hand side, we have terms involving spatial derivatives of $\p_0^J\rho$, which we can replace with an expansion in $\theta$ by using \eq{rho_exp_new} with $\ell=L$ and $j=J$, with an error which is $O_{\infty}(M^{-L-3+J})$. We also have a term involving $\p_0^{J+2}\rho$, which we can replace with an expansion in $\theta$ by using \eq{rho_exp_new} with $\ell=L+1$ and $j=J+2$ to give an error which is $O_{\infty}(M^{-L-2+J})$. 

Expanding the third term in \eq{dJrho}, we will have a term of the form $M^{-2}\p_0^{J+1}\rho \p_0\rho$, which we can replace with an expansion in $\theta$ by using \eq{rho_exp_new} with $\ell=L+1$ and $j=J+1$ to give an error which is $O_{\infty}(M^{-L-3+J})$. We also have terms of the form $M^{-2}\p_0^{K}\rho \p_0^M\rho$ for $K\leq M \leq J$, which we can replace using  \eq{rho_exp_new} with $\ell=L$ and $j=K, M$ to give an error which is $O_{\infty}(M^{-L-3+M}) = O_{\infty}(M^{-L-3+J})$. Similarly expanding the derivatives in the final term of \eq{rho_exp_new}, we can use  \eq{rho_exp_new} with $\ell=L$ to replace all of the terms that appear with an expansion in $\theta$ with error $O_{\infty}(M^{-L-3+J})$. Thus, we have established an expansion of the form we require exists for $\ell=L+1, j=J$. By examining the substitutions that we make, we see that each term in the expansion for $\ell=L+1, j=J$ must be the $x^0$ derivative of the corresponding term in the expansion for $\ell=L, j=J-1$, so our result holds for $\ell = L+1, j=J$. We can thus work backwards to $j=1$ by induction. At $j=0$ we again find an expansion, but a new term may be generated. This new term must be a polynomial in the lower order terms, and their derivatives up to order $2$, which recovers the induction assumption that $\mathcal{F}_{2k}$ is a polynomial in $\theta$ and its derivatives up to order $2k-1$. This closes the induction on $l$.

Now assume that \eq{rho_exp_new} holds. This immediately implies that the bounds on $\rho$ in \eq{Bded} must hold by taking $\ell = j$ for $j = 0, \ldots, l$ and also $j=l+1$, $\ell = l$. To obtain the bounds on $\theta$, we recall that
\[
\p_0^2\theta = - \p_i \p_i \theta + 2 \p^\mu \rho \p_\mu \theta.
\]
Since we assume $\theta, \p_0 \theta$ are $O_\infty(1)$, we deduce $\p_0^2\theta$ is $O_\infty(1)$. Differentiating the equation with respect to $\p_0$, we can inductively recover the bounds on $\p_0^j \theta$.
\end{proof}
We note that the condition on $\theta$ in \eq{Bded} actually follows from the conditions on $\rho$, together with the assumptions on $\theta, \p_0 \theta$ but for simplicity we include it in the hypotheses.

\subsection{Proof of Theorem \ref{thm:EFT_sol}}\label{pf:EFT_sol}

To prove Theorem \ref{thm:EFT_sol}, we require:
\begin{Lemma} \label{EFTEnEst}
Let $\psi,g$ be smooth functions of $x^\mu \in \mathcal{M}_T$ and $M$ satisfying $-\Box \psi + m^2 \psi = g$ (with $m,T$ independent of $M$) with $\psi|_{x^0}=\psi_0$ and $\partial_0 \psi |_{x^0}=\psi_1$. Assume that, as $M \rightarrow \infty$, $\p_0^j g= O_{\infty}(M^{-l})$ for all $0\leq j \leq p$ and that $\psi_0,\psi_1$ and all of their derivatives, are also $O_{\infty}(M^{-l})$. Then $\p_0^j \psi = O_{\infty}(M^{-l})$ for all $0\leq j \leq p+1$. In particular if $\p_0^j g= O_{\infty}(M^{-l})$ for all $j\geq 0$, then all derivatives of $\psi$ are $O(M^{-l})$.
\end{Lemma}
\begin{proof} 
The case $j=0$ follows from Lemma \ref{EnEst}, together with Morrey's inequality \eq{Morrey} below. Now suppose the result holds for $p=P$. Differentiating the equation, we have:
\[
-\Box \psi' + m^2  \psi' = \p_0 g,
\]
where $\psi' = \p_0 \psi$. Clearly $ \psi' |_{x^0} = \partial_0 \psi |_{x^0}=\psi_1 = O_{\infty}(M^{-l})$. We can rearrange the equation to give $\p_0 \p_0 \psi = \p_i \p_i \psi - m^2 \psi - g$, so that $\p_0 \psi' |_{x^0} = \partial_0 \p_0 \psi |_{x^0}=\p_i \p_i \psi_0 - m^2 \psi_0 - g|_{\{x^0=0\}} = O_{\infty}(M^{-l})$. Thus we can apply the result for $p=P$ to $\psi'=\p_0 \psi$, which shows that the result must hold for $p=P+1$.
 \end{proof}

To prove the ``existence" part of Theorem \ref{thm:EFT_sol} we define functions $\theta^{(k)}$, $0 \le k \le l$, on $\mathcal{M}_T$ as follows. Substitute
\ben{thetaexp}
 \theta = \sum_{k=0}^l \frac{\theta^{(k)}}{M^k}
\een
into \eqref{genEFT} and formally equate coefficients of $M^{-1}$ up to $M^{-l}$. At leading order we obtain
\ben{theta0eq}
 -\Box \theta^{(0)} + m^2 \theta^{(0)} = 0
\een
and at order $M^{-k}$ ($k \le l$) we obtain
\ben{thetakeq}
 -\Box \theta^{(k)} + m^2 \theta^{(k)} = P_k
\een
where $P_k$ is a polynomial in $\theta^{(0)}, \ldots, \theta^{(k-1)}$ and their derivatives, with coefficients independent of $M$. We now define $\theta^{(0)}$ as the solution of the Klein-Gordon (or wave) equation \eqref{theta0eq} on $\mathcal{M}_T$ with initial data
\be
 \theta^{(0)} |_{x^0=0} = \theta_0 \qquad \qquad \partial_0 \theta^{(0)} |_{x^0=0} = \theta_1
 \ee
Since all derivatives of $\theta_0$ and $\theta_1$ are uniformly bounded as $M \rightarrow \infty$, Lemma \ref{EFTEnEst} with $l=0$ shows that all derivatives of $\theta^{(0)}$ are uniformly bounded on $\mathcal{M}_T$ as $M \rightarrow \infty$ (with fixed $T$). 

Now define $\theta^{(k)}$ for $1 \le k \le l$ as the solution of the inhomogeneous Klein-Gordon (or wave) equation \eqref{thetakeq} on $\mathcal{M}_T$ with trivial initial data
\be
 \theta^{(k)} |_{x^0=0} = \partial_0 \theta^{(k)} |_{x^0=0} = 0
 \ee
Since $P_1$ is a polynomial in $\theta^{(0)}$ and its derivatives, it is uniformly bounded as $M \rightarrow \infty$. So we can again use standard energy estimates to deduce that $\theta^{(1)}$ and its derivatives are uniformly bounded as $M \rightarrow \infty$. Proceeding inductively we deduce that $\theta^{(k)}$ and its derivatives are uniformly bounded as $M \rightarrow \infty$. 

We now define $\theta$ by \eqref{thetaexp}. Clearly $\theta$ and its derivatives are uniformly bounded on $\mathcal{M}_T$ as $M \rightarrow \infty$, and $\theta$ satisfies the desired initial conditions. 

Finally we use equation \eqref{genEFT} to define $R_l$. Using the Klein-Gordon equations for $\theta^{(k)}$ this gives
\be
 \frac{R_l}{M^{l+1}} = \sum_{k=1}^l \frac{(P_k - S_k(\theta,\partial\theta,\ldots))}{M^k}
\ee 
The RHS is a polynomial in $M^{-1}$ with coefficients that are polynomials in the $\theta^{(k)}$ and their derivatives. From the definition of the $P_k$, the coefficients of $M^{-1}, M^{-2}, \ldots, M^{-l}$ in this polynomial must vanish, so the first term with non-vanishing coefficient is $M^{-l-1}$. Hence $R_l$ is a polynomial in $M^{-1}$ with coefficients that are polynomials in the $\theta^{(k)}$ and their derivatives. This implies that $R_l$ and all of its derivatives are uniformly bounded as $M \rightarrow \infty$. Hence $\theta$ is an EFT solution to order $l$. 

Now we prove the ``uniqueness" part of the theorem.  Let $\theta$ and $\tilde{\theta}$ be two EFT solutions to order $l$, with corresponding functions $R_l$ and $\tilde{R}_{l}$. Define $\delta \theta = \tilde{\theta} - \theta$. We then have
\ben{deltathetaeq}
  -\Box \delta\theta + m^2\delta  \theta = \sum_{k=1}^l \frac{[S_k]}{M^k} + \frac{[R_l]}{M^{l+1}}
\een
where
\be
[S_k] = S_k(\tilde{\theta},\partial\tilde{\theta},\ldots)-S_k(\theta,\partial\theta,\ldots) \qquad [R_l] = \tilde{R}_l - R_l
\ee
Since $\theta$ and $\tilde{\theta}$ satisfy the same initial conditions we have
\be
 \delta \theta |_{x^0=0} = \partial_0 \delta\theta |_{x^0=0} = 0
\ee
$[S_k]$ is a polynomial in $\theta$, $\tilde{\theta}$ and their derivatives, with coefficients independent of $M$. 
Uniform boundedness of $\theta$, $\tilde{\theta}$ and their derivatives implies that $\p_0^j [S_k] = O_{\infty}(1)$  for all $j \geq 0$. Hence the RHS of \eqref{deltathetaeq}, and all of its derivatives, are $O(M^{-1})$. Furthermore our initial data are trivially $O_{\infty}(M^{-1})$. Hence we can apply the Lemma to deduce that $\p_0^j \delta \theta= O_{\infty}(M^{-1})$ for all $j \geq 0$. 

If $l=0$ then we are done. For $l \ge 1$ we now consider $\partial_0^q [ S_k]$. Writing $\tilde \theta = \theta + \delta \theta$, this is a polynomial in $\theta,\delta \theta$ and their derivatives, with coefficients independent of $M$. Furthermore, each term in the polynomial contains at least one factor of $\delta \theta$ or a derivative of $\delta \theta$. Each such term is $O_{
\infty}(M^{-1})$ using the result just established, together with the boundedness of $\theta$ and its derivatives. Hence  $\partial_0^q [ S_k]=O_{\infty}(M^{-1})$. This implies that the RHS of \eqref{deltathetaeq}, and all of its derivatives, are actually $O(M^{-2})$. So our Lemma lets us improve our estimate to deduce that $\delta \theta$ and all of its derivatives are $O(M^{-2})$. 

If $l=1$ then we are done. If $l>1$ then we can continue: using $\partial_0^r \delta \theta = O_{\infty}(M^{-2})$ and our boundedness assumptions gives $\partial_0^q [S_k]=O_{\infty}(M^{-2})$. Hence the RHS of \eqref{deltathetaeq} is $O(M^{-3})$ so our Lemma implies that $\delta \theta$ and all of its derivatives are $O(M^{-3})$. We  can then repeat this process, at each stage improving our estimate of the RHS of \eqref{deltathetaeq} (and its derivatives), until we reach $O(M^{-l-1})$ at which point we cannot do any better because of the final term in \eqref{deltathetaeq}. Hence $\delta \theta$ and all of its derivatives are $O(M^{-l-1})$.

Our existence and uniqueness results above constitute two of the usual three requirements of well-posedness for the Cauchy problem for \eq{genEFT}, suitably modified to account for the fact that our equation only holds to a certain order in $1/M$. Typically one would also ask that the solution depends continuously on initial data. To establish such a result in our problem, we need to specify a topology on the space of EFT solutions. We specialise to the case where $\theta_0, \theta_1$ are independent of $M$ (hitherto we have only required boundedness in $M$). Under this assumption, we observe that our uniqueness result above, together with the proof of the existence result, implies that any EFT solution to \eq{genEFT} must take the form
\[
 \theta = \sum_{k=0}^l \frac{\theta^{(k)}}{M^k} + \frac{R_{l}}{M^{l+1}},
\]
with $R_l$ and all its derivatives uniformly bounded, and $\theta^{(k)}$ independent of $M$ and uniquely determined by solving \eq{thetakeq}. We can define a topology on the EFT solutions by identifying $\theta$ with the vector $(\theta^{(k)})_{k=0}^l$, thought of as an element of $C^\infty(\mathcal{M}_T)$ with its usual topology. Assuming we consider our initial data $\theta_0, \theta_1$ to belong to $C^\infty(\T^n)$ with its usual topology, then it is straightforward to see that $\theta$ depends continuously on the initial data.

\subsection{Causality}

Now we discuss briefly the topic of causality in EFT solutions. Specifically, for our EFT for $\theta$, does the effective metric $(G^{-1})_{\mu\nu}$ of \eqref{nonlinear_wave} provide a better description of causality than the Minkowski metric $g_{\mu\nu}$? 

Consider a solution $\theta$ of \eqref{nonlinear_wave} and deform its initial data $\theta_0,\theta_1$ in some compact subset $\Omega$ of $\{x^0=0\}$ to obtain new initial data $\theta_0',\theta_1'$ that agree with the original data outside $\Omega$. Let $\theta'$ be the corresponding solution of \eqref{nonlinear_wave}. Then a standard domain of dependence argument implies that $\theta'-\theta$ vanishes outside the hypersurface ${\mathcal N}_G$ that is null w.r.t. $(G^{-1})_{\mu\nu}$ and emanates outwards orthogonally (w.r.t.$(G^{-1})_{\mu\nu}$) from $\partial \Omega$.\footnote{More precisely, ${\mathcal N}_G$ is the Cauchy horizon, defined w.r.t. $(G^{-1})_{\mu\nu}$, of the complement of $\Omega$.} 

Now let $\tilde{\theta}$ be any other EFT solution of order $3$, with initial data $\theta_0,\theta_1$ and $\tilde{\theta}'$ any other EFT solution of order $3$ with the deformed initial data $\theta_0',\theta_1'$. From Theorem \eqref{thm:EFT_sol} we know that $\tilde{\theta} - \theta$ and $\tilde{\theta}' - \theta'$, and their derivatives, are $O(M^{-4})$ everywhere. But we can write $\tilde{\theta}' - \tilde{\theta} = (\tilde{\theta}' - \theta')+ (\theta' - \theta) + (\theta - \tilde{\theta}) = \theta'-\theta + O(M^{-4})$ and hence $\tilde{\theta}' - \tilde{\theta}$ is $O(M^{-4})$ outside ${\mathcal N}_G$ (and so are its derivatives). 

Causal properties of equation \eqref{nonlinear_wave} have let us understand causal properties of any EFT solution. But we can deduce exactly the same result without using \eqref{nonlinear_wave}. For the above initial data, let $\theta,\theta'$ be the EFT solutions  of order $l$ constructed using the perturbative method used in the proof of Theorem \ref{thm:EFT_sol} (i.e. expanding in powers of $M^{-2}$). In this case, at each step of the iteration one solves a standard wave equation with principal symbol $g^{\mu\nu}$ and so $\theta,\theta'$ agree outside ${\mathcal N}_g$, the hypersurface that is null w.r.t. $g_{\mu\nu}$ emanating outwards orthogonally from $\partial \Omega$. This lies outside ${\mathcal N}_G$. However, in the region between ${\mathcal N}_G$ and ${\mathcal N}_g$, $\theta'-\theta$ is negligibly small. To see this, fix $C,\alpha>0$ and define ${\mathcal V}(M) = \{ (x^0,x^i): \exists (x^0,y^i) \in {\mathcal N}_g \,\, s.t. \,,\, ||x^i-y^i||<CM^{-\alpha} \}$, i.e., a region within distance $CM^{-\alpha}$ of ${\mathcal N}_g$. Since the speed of light w.r.t. $(G^{-1})_{\mu\nu}$ differs from $1$ by $O(M^{-2})$ we expect ${\mathcal N}_G$ to lie in the region ${\mathcal V}(M)$ for some $C$ if $\alpha< 2$. Now, given a point $(x^0,x^i) \in {\mathcal V}(M)$ we can Taylor expand $\theta'-\theta$ to order $k$ around a point $(x^0,y^i)$ of ${\mathcal N}_g$. On ${\mathcal N}_g$, all derivatives of $\theta'-\theta$ vanish and so only the ``error" term contributes to the Taylor expansion. Boundedness of the derivatives of our EFT solution implies that this error term is $O(M^{-(k+1)\alpha})$. Since $k$ is arbitrary, this implies that $\theta'-\theta$ vanishes faster than any power of $M^{-1}$ in ${\mathcal V}(M)$. Theorem \ref{thm:EFT_sol} now implies that any other EFT solution of order $l$ is $O(M^{-l-1})$ in ${\mathcal V}(M)$ (and outside ${\mathcal N}_g$). In particular any EFT solution of order $3$ is $O(M^{-4})$ outside ${\mathcal N}_G$, as deduced from \eqref{nonlinear_wave}. 

We could consider a different EFT in which we reverse the sign of ${\mathcal F}_2$, with a corresponding reversal of sign in the second term of $G^{\mu\nu}$ in \eqref{nonlinear_wave}. In this case, the null cone of $(G^{-1})_{\mu\nu}$ lies {\it outside}, or on, the null cone of $g_{\mu\nu}$ so in this case \eqref{nonlinear_wave} permits superluminal propagation  \cite{Adams:2006sv}, i.e., ${\mathcal N}_G$ lies outside ${\mathcal N}_g$. But arguing as above with the roles of $g_{\mu\nu}$ and $(G^{-1})_{\mu\nu}$ interchanged one can see that \eqref{nonlinear_wave} implies that $\theta'-\theta$ is negligibly small in the region between ${\mathcal N}_g$ and ${\mathcal N}_G$. Thus the superluminality is not observable within the framework of EFT solutions as we have defined them.

The above discussion refers to the solution defined on a time interval $T$ that does not depend on $M$. If we consider a long time interval $T \propto M^2$ then the difference between ${\mathcal N}_g$ and ${\mathcal N}_G$ would be $O(1)$ as $M \rightarrow \infty$ and one might expect \eqref{nonlinear_wave} to provide the better description of causality. However, we will see in section \ref{LTS} that we have no reason to expect a solution of \eqref{nonlinear_wave} to approximate the UV solution over such a long time interval. (We will also see that there is no reason to expect convergence of the perturbative approach when $T \propto M^2$.)

\section{The initial value problem} \label{IVP}

\subsection{Preliminaries} \label{normdef}We first introduce some notation. We say that a multi-index $A = (A_1, \ldots, A_n) \in \N^{n}_0$ is an $n$-tuple of non-negative integers and we let $\abs{A} = A_1 + \cdots + A_n$. We can compactly write higher-order spatial derivatives as:
\[
D^A := (\p_1)^{A_1} \cdots (\p_n)^{A_n}.
\]
For a function $f \in C^\infty(\T^n)$ and $k \in \N_0$, we denote:
\[
\norm{f}{C^k} = \sup_{\substack{\abs{A} \leq k \\ x \in \T^n}} \abs{D^Af(x)}, \qquad \norm{f}{\dot{C}^k} = \sup_{\substack{\abs{A} = k \\ x \in \T^n}} \abs{D^Af(x)}.
\]
and introduce the Sobolev norms:
\[
\norm{f}{H^k} = \left( \sum_{\abs{A}\leq k} \norm{D^A f }{L^2}^2 \right)^{\frac{1}{2}}, \qquad \norm{f}{\dot{H}^k} = \left( \sum_{\abs{A}= k} \norm{D^A f }{L^2}^2 \right)^{\frac{1}{2}}.
\]
We note Morrey's inequality, which states that if $k, l \in  \N_0$ and $k > \frac{n}{2}+l$, then 
\ben{Morrey}
\norm{f}{C^l} \leq c \norm{f}{H^k},
\een
for some constant $c$ depending on $k, l, n, L$ only\footnote{Here and subsequently, $c$ will denote a universal constant, independent of $M$ and initial data which may change from line to line.}.

For $u \in C^\infty(\mathcal{M}_T)$, $l \in \N_0$ and $X^k$ standing for one of $C^k$, $\dot{C}^k$ $H^k$ or $\dot{H}^k$, we introduce the norms:
\[
\norm{u}{\dot{C}^l_t X_x^k} := \sup_{ t \in (0, T)} \norm{\p_0^l u(t, \cdot)}{X^k}, \qquad \norm{u}{C^l_t X_x^k} = \max_{j \leq l} \norm{u}{\dot{C}^j_t X_x^k}.
\]
Finally for $1\leq p <\infty$ we define:
\[
\norm{u}{L^p_t X_x^k} := \left(\int_0^T \norm{u(t, \cdot)}{X^k}^p dt \right)^{\frac{1}{p}}.
\]

We recall the following standard result:
\begin{Lemma}\label{EnEst}
Suppose $\mu \geq 0$,  $u_0, u_1 \in C^\infty(\T^n)$ and $f \in C^\infty(\overline{\mathcal{M}_T})$. Then there exists a unique $u \in C^\infty(\overline{\mathcal{M}_T})$ solving:
\[
\begin{array}{c}
\displaystyle- \Box u + \mu^2 u= f \\[.4cm]
u|_{x^0=0} = u_0, \qquad \p_0 u|_{x^0=0} = u_1.
\end{array}
\]
Furthermore, for any $k$, we have the estimate:
\begin{align*}
& \left( \norm{u}{\dot{C}^{1}_t \dot{H}^k_x}^2 + \norm{u}{C^0_t \dot{H}^{k+1}_x }^2+ \mu^2 \norm{u}{C^0_t \dot{H}^k_x}^2 \right)^{\frac{1}{2}} \leq \norm{f }{L^1_t \dot{H}^k_x}\\&\qquad + \left( \norm{u_0}{\dot{H}^{k+1}}^2 + \norm{u_1}{\dot{H}^k }^2+ \mu^2 \norm{u_0}{\dot{H}^k}^2 \right)^{\frac{1}{2}}.
\end{align*}
\end{Lemma}

\subsection{Bounding existence time} Recall that we are interested in solutions to 
\begin{align}
\displaystyle- \Box \rho + M^2 \rho &= \displaystyle \p_\mu \rho \p^\mu \rho - \p_\mu \theta  \p^\mu \theta- M^2 \rho^2 W(\rho), \label{rhothetaeq12}\\[.1cm]
- \Box \theta  &= \ 2 \p^\mu \rho \p_\mu \theta, \label{rhothetaeq22}
\end{align}
subject to:
\ben{initcond}
\rho|_{\{x^0=0\}} = \rho_0, \quad \p_0 \rho|_{\{x^0=0\}} = \rho_1, \quad \theta|_{\{x^0=0\}} = \theta_0, \quad \p_0 \theta|_{\{x^0=0\}} = \theta_1. 
\een
We assume that $\rho_0, \rho_1, \theta_0, \theta_1$ are given, smooth, functions on $\T^n$, which may depend on $M$, but will later be assumed to satisfy some uniform bounds. Standard results concerning the local well-posedness of hyperbolic equations (see for example \cite{sogge1995lectures, ringstrom, luk}) give:

\begin{Theorem}\label{LocEx}
Let $k_0 = 2 \lfloor \frac{n}{2} \rfloor + 1$ and set:
\[
d =  \norm{ \theta_0}{H^{k_0+1}} + \norm{\theta_1}{H^{k_0}} +  \norm{ \rho_0}{H^{k_0+1}} + \norm{\rho_1}{H^{k_0}} 
\]
Then a unique smooth solution to  \eq{rhothetaeq12}, \eq{rhothetaeq22} subject to \eq{initcond} exists on some time interval $[0, T)$, where $T=T(M, d) >0$. Further, if $[0, T_{max})$ is the \emph{maximal} time on which a smooth solution exists, either $T_{max}=\infty$ or else there exists a sequence $t_k \to T_{max}$ such that
\[
\norm{ \theta(t_k, \cdot)}{H^{k_0+1}} + \norm{\p_0\theta(t_k, \cdot)}{H^{k_0}} +  \norm{ \rho(t_k, \cdot)}{H^{k_0+1}} + \norm{\p_0\rho(t_k, \cdot)}{H^{k_0}} \to \infty.
\]
\end{Theorem}

A priori, $T_{max}$ will depend on $M$. Our goal is to show that, provided we control the initial data appropriately, we can bound $T_{max}$ from below, independently of $M$. Our argument will be based on a bootstrap. Let us introduce
\begin{align*}
N(t) &:= \left(\sum_{i=1}^n \norm{ \p_i \theta(t, \cdot)}{H^{k_0}}^2 + \norm{\p_0 \theta(t, \cdot)}{H^{k_0}}^2 +  \norm{ \rho(t, \cdot)}{H^{k_0+1}}^2 + M^2 \norm{ \rho(t, \cdot)}{H^{k_0}}^2 + \norm{\p_0\rho(t, \cdot)}{H^{k_0}}^2\right)^{\frac{1}{2}}, \\
N &= N(0)= \left( \sum_{i=1}^n \norm{ \p_i \theta_0}{H^{k_0}}^2 + \norm{\theta_1}{H^{k_0}}^2 +  \norm{ \rho_0}{H^{k_0+1}}^2 + M^2 \norm{ \rho_0}{H^{k_0}}^2 + \norm{\rho_1}{H^{k_0}}^2\right)^{\frac{1}{2}}.
\end{align*}
We define $T^* = \sup\{ t | N(s) \leq 10 N \textrm{ for all }s\leq t \}$. We observe that by \eq{LocEx} the solution exists and remains smooth on $[0, T^*)$. If we can show that $T^*$ is bounded below by a positive quantity depending only on $N$, then we will be done.

\begin{Lemma}
Suppose $\sup_{t \in (0, T)} N(t) \leq 10N$. Then we can estimate:
\[
\sup_{t \in (0, T)} N(t) \leq \sqrt{2} N + cT N^2  (1+ e^{3 N}),
\]
for a constant $c$ that depends only on $n, L$
\end{Lemma}
\begin{proof}
We first consider  \eq{rhothetaeq22}. Let us assume $k \leq k_0$. Applying Lemma \ref{EnEst} on some time interval $(0, T)$:
\[
\norm{\theta}{\dot{C}^{1}_t \dot{H}^k_x}^2 + \norm{\theta}{C^0_t \dot{H}^{k+1}_x }^2 \leq 8 \norm{\p^\mu \rho \p_\mu \theta}{L^1_t \dot{H}^k_x}^2 + 2\norm{\theta_0}{\dot{H}^{k+1}}^2 + 2\norm{\theta_1}{\dot{H}^{k}}^2,
\]
where the spacetime norms refer to $(0, T) \times \T^n$. We consider the first term on the right-hand side. For $\abs{A} = k$ we can write:
\[
D^A \left(\p^\mu \rho \p_\mu \theta \right) = \sum_{B+C = A} c_{B} \p^\mu D^B \rho \p_\mu D^C \theta
\]
for some combinatorial constants $c_B$. Now, we claim that either $\abs{B}< k_0 - \frac{n}{2}$ or $\abs{C}< k_0 - \frac{n}{2}$ holds. If not, then $k_0 \geq k = \abs{B}+\abs{C}   \geq 2k_0 - 2 \lfloor \frac{n}{2} \rfloor$ which implies $k_0 \leq 2 \lfloor \frac{n}{2} \rfloor$, which is false. We deduce, by Morrey's inequality that either:
\[
\norm{\p^\mu D^B \rho (t, \cdot)}{C^0} \leq c' N(t), \quad \textrm{or} \quad \norm{\p_\mu D^C \theta(t, \cdot)}{C^0} \leq c' N(t)
\]
holds, where $c'$ is a universal constant depending on $k, k_0, n$. By virtue of the estimate:
\[
\norm{fg}{L^2} \leq \norm{f}{C^0} \norm{g}{L^2},
\]
we can bound
\[
\norm{D^A \left(\p^\mu \rho \p_\mu \theta \right)}{L^2} \leq \sum_{B+C = A} |c_{B}| \norm{\p^\mu D^B \rho \p_\mu D^C \theta}{L^2} \leq Cc'^2 N^2,
\]
for a combinatorial constant $C$. We conclude that there exists a constant $c$, depending on $k, k_0, n$ such that:
\[
\norm{\p^\mu \rho \p_\mu \theta}{L^1_t \dot{H}^k_x} \leq c T N^2.
\]

Next, we consider \eq{rhothetaeq12}. Applying Lemma \ref{EnEst}, we obtain for $k \leq k_0$:
\begin{align}
&\norm{\rho}{\dot{C}^{1}_t \dot{H}^k_x}^2 + \norm{\rho}{C^0_t \dot{H}^{k+1}_x }^2+ M^2 \norm{\rho}{C^0_t \dot{H}^{k}_x }^2 \leq \nonumber \\
&\qquad \qquad 2\norm{\rho_0}{\dot{H}^{k+1}}^2 +2M^2\norm{\rho_0}{\dot{H}^{k}}^2 + 2\norm{\rho_1}{\dot{H}^{k}}^2\label{rhoest}\\
&\qquad \qquad + 6\norm{\p^\mu \rho \p_\mu \rho}{L^1_t \dot{H}^k_x}^2+6\norm{\p^\mu \theta \p_\mu \theta}{L^1_t \dot{H}^k_x}^2+6 \norm{M^2 \rho^2 W(\rho)}{L^1_t \dot{H}^k_x}^2.\nonumber
\end{align}
Here we recall that $W:\rho \mapsto (e^{2 \rho}-1 - 2 \rho) /(2\rho^2)$ satisfies, for each $l$:
\[
\abs{W^{(l)}(\rho)} \leq c_l e^{2 \abs{\rho}}
\]
for some constants $c_l$. The first two terms on the final line of \eq{rhoest} may be dealt with precisely as for the nonlinearity in the $\theta$ equation to give:
\[
 \norm{\p^\mu \rho \p_\mu \rho}{L^1_t \dot{H}^k_x}+\norm{\p^\mu \theta \p_\mu \theta}{L^1_t \dot{H}^k_x} \leq c T N^2.
\]
For the final term in \eq{rhoest}, we note that if $\abs{A} = k$ we can write:
\[
D^A(M^2 \rho^2 W(\rho)) = \sum_{B + C + E = A} c_{B, C} D^B (M\rho) D^C (M \rho) D^E(W(\rho))
\]
Now, by repeated application of the chain and product rules, we may expand $D^E(W(\rho))$ as a sum of terms, with universal coefficients, of the form
\[
W^{(l)}(\rho) \times D^{E^1} \rho \times \cdots D^{E^l} \rho
\]
where $l \leq \abs{E}$ and the multi-indices $E^i$ satisfy $E^1+\ldots +E^l = E$. This gives us an expression:
\[
D^A(M^2 \rho^2 W(\rho)) = \sum_{\substack{l \leq \abs{A} \\ B + C + E^1+\cdots+ E^l = A}} c_{B, C, E^1, \ldots, E^l} W^{(l)}(\rho)D^B (M\rho) D^C (M \rho) D^{E^1} \rho \cdots D^{E^l} \rho
\]
where $c_{B, C, E^1, \ldots, E^l}$ are universal constants.  By a similar argument to above, all the derivatives of $\rho$ on the RHS have order less than $k_0-\frac{n}{2}$, except for possibly one which we can estimate in $L^2$. We can estimate the terms individually as:
\[
\norm{W^{(l)}(\rho)D^B (M\rho) D^C (M \rho) D^{E^1} \rho \cdots D^{E^l} \rho}{L^2} \leq c' N^{2+l} e^{2N} \leq c_l N^2 e^{3N}
\]
so that combining these estimates we find
\[
\norm{M^2 \rho^2 W(\rho)}{L^1_t \dot{H}^k_x} \leq cT N^2e^{3N}
\]
for some universal constant $c$.

Putting together our estimates, we have:
\begin{align*}
&\norm{\theta}{\dot{C}^{1}_t \dot{H}^k_x}^2 + \norm{\theta}{C^0_t \dot{H}^{k+1}_x }^2+\norm{\rho}{\dot{C}^{1}_t \dot{H}^k_x}^2 + \norm{\rho}{C^0_t \dot{H}^{k+1}_x }^2+ M^2 \norm{\rho}{C^0_t \dot{H}^{k}_x }^2 &
\\&  \qquad \qquad \leq 2\norm{\theta_0}{\dot{H}^{k+1}}^2 + 2\norm{\theta_1}{\dot{H}^{k}}^2+2\norm{\rho_0}{\dot{H}^{k+1}}^2 +2M^2\norm{\rho_0}{\dot{H}^{k}}^2 + 2\norm{\rho_1}{\dot{H}^{k}}^2 \\
& \qquad \qquad \qquad + c(T^2N^4 + N^4T^2 e^{6N}).
\end{align*}
Summing over $k=0, \ldots, k_0$ and taking a square root, we conclude:
\[
\sup_{t \in (0, T)} N(t) \leq \sqrt{2} N + cT N^2  (1+ e^{3 N}). \qedhere
\]
\end{proof}
This immediately gives us the first part of Proposition \ref{Prop1}:
\begin{Theorem}\label{UnifEx}
There exists a universal constant $c>0$ such that
\[
T^* \geq \frac{c}{N(1+e^{3N})}.
\]
In particular, provided the initial data is such that $N$ is bounded independently of $M$, a unique smooth solution to \eq{rhothetaeq12}, \eq{rhothetaeq22} subject to \eq{initcond} exists on some time interval $[0, T)$, where $T>0$ is independent of $M$.
\end{Theorem}
\begin{proof}
Suppose not. Then for any $\epsilon>0$ we can find data for which
\[
T^* N(1+e^{3N}) < \epsilon.
\]
Taking $\epsilon$ sufficiently small, by the previous result we have that:
\[
\sup_{t\in (0, T^*)} N(t)  \leq 9N,
\]
but since $N(t)$ is a continuous function of its argument, this contradicts the definition of $T^*$.
\end{proof}

\subsection{Propagating bounds}
Having established a uniform time of existence for our solutions, we next turn to showing uniform bounds on derivatives of the solution.

\begin{Theorem}
Let $k_0$ and $T^*$ be as above. For each $k \geq k_0$ let:
\begin{align*}
N_k(t) &:= \left(\sum_{i=1}^n \norm{ \p_i \theta(t, \cdot)}{H^{k}}^2 + \norm{\p_0 \theta(t, \cdot)}{H^{k}}^2 +  \norm{ \rho(t, \cdot)}{H^{k+1}}^2 + M^2 \norm{ \rho(t, \cdot)}{H^{k}}^2 + \norm{\p_0\rho(t, \cdot)}{H^{k}}^2\right)^{\frac{1}{2}}, \\
N_k &= N_k(0)= \left( \sum_{i=1}^n \norm{ \p_i \theta_0}{H^{k}}^2 + \norm{\theta_1}{H^{k}}^2 +  \norm{ \rho_0}{H^{k+1}}^2 + M^2 \norm{ \rho_0}{H^{k}}^2 + \norm{\rho_1}{H^{k}}^2\right)^{\frac{1}{2}}.
\end{align*}
Then we have the bound:
\[
\sup_{t \in (0, T^*)} N_k(t) \leq c_k
\]
for some constant $c_k$ depending on the initial data only through $N_k$.
\end{Theorem}
\begin{proof}
We work by induction. For $k=k_0$, the base case, this follows immediately from the definition of $T^*$. Suppose that we have established the bound for $k-1$ where $k>k_0$. Considering \eq{rhothetaeq22} and applying Lemma \ref{EnEst} on a time interval $(0, T)$ for $T<T^*$, we obtain:
\[
\norm{\theta}{\dot{C}^{1}_t \dot{H}^k_x}^2 + \norm{\theta}{C^0_t \dot{H}^{k+1}_x }^2 \leq 8 \norm{\p^\mu \rho \p_\mu \theta}{L^1_t \dot{H}^k_x}^2 + 2\norm{\theta_0}{\dot{H}^{k+1}}^2 + 2\norm{\theta_1}{\dot{H}^{k}}^2,
\]
As in the proof of the previous theorem, for $\abs{A} = k$ we can write:
\[
D^A \left(\p^\mu \rho \p_\mu \theta \right) = \sum_{B+C = A} c_{B} \p^\mu D^B \rho \p_\mu D^C \theta
\]
for some combinatorial constants $c_B$. Arguing as before, we see that either $\abs{B}< k -1 - \frac{n}{2}$ or $\abs{C}< k - 1-\frac{n}{2}$ holds. Thus we can bound one of the factors of $\p^\mu D^B \rho \p_\mu D^C \theta$ in $C^0$ by $N_{k-1}(t)$ and the other we bound in $L^2$. Since $N_{k-1}\leq N_k$, by the induction assumption we see that we can bound $N_{k-1}(t)$ by some constant $c_{N_k}$ depending only on $N_k$. We thus have
\[
\norm{\p^\mu \rho \p_\mu \theta}{L^1_t \dot{H}^k_x} \leq c_{N_{k}} T \sup_{t \in (0, T)} N_k(t).
\]

We can argue similarly with \eq{rhothetaeq12} and then combine the estimates as in the previous Theorem to obtain:
\[
\sup_{t \in (0, T)} N_k(t) \leq \sqrt{2}N_k + c_{N_{k}} T  \sup_{t \in (0, T)} N_k(t).
\]
for some constant $c_{N_{k}}$. Taking $T = 1/(2c_{N_{k}})$ we obtain:
\[
\sup_{t \in (0, T)} N_k(t) \leq 2\sqrt{2}N_k
\]
and we can iterate this estimate on the intervals $(T, 2T), \ldots, ((m-1)T, mT)$ to obtain:
\[
\sup_{t \in (0, T^*)} N_k(t)  \leq (2\sqrt{2})^m N_k
\]
where $m$ is chosen such that $m/(2c_{N_{k}})>T^*$, which gives the result we require and so we are done by induction.
\end{proof}

This result allows us to propagate bounds on spatial derivatives of the initial data. Assuming $k >p +\frac{n}{2}$ we may use Morrey's inequality to establish control of $\norm{\theta}{C^0_tC^{p}_x}$, $\norm{\p_0\theta}{C^0C^p_x}$, $M\norm{\rho}{C^0_tC^p_x}$ and  $\norm{\rho}{C^0_tC^p_x}$ in terms of $N_k(t)$ and hence $N_k$. To bound $\norm{\p_0\p_0\theta}{C^0C^p_x}$, we rearrange the equation for $\theta$ to write $\p_0\p_0\theta$ in terms of spatial derivatives of $\theta, \rho, \p_0\rho$ which we have already bounded. This completes the proof of Proposition \ref{Prop1}. In order to establish Proposition \ref{Prop2} we wish to propagate also time derivatives of the initial data.

\begin{Theorem} \label{tThm}
Let $k_0$ and $T^*$ be as above. For each $l \geq 0$ and $k \geq k_0+l$ let :
\begin{align*}
N_{k,l}(t) &:= \Bigg[ \sum_{j=0}^l \bigg ( \sum_{i=1}^n \norm{ \p_0^j\p_i \theta(t, \cdot)}{H^{k-j}}^2+ \norm{\p_0^{j+1} \theta_1(t, \cdot)}{H^{k-j}}^2 +  \norm{ \p_0^j\rho(t, \cdot)}{H^{k-j+1}}^2 +\\
&\qquad \qquad  + M^2 \norm{ \p_0^j \rho(t, \cdot)}{H^{k-j}}^2 + \norm{\p_0^{j+1}\rho(t, \cdot)}{H^{k-j}}^2 \bigg ) \Bigg]^{\frac{1}{2}}\\
N_{k,l} &= N_{k,l}(0).
\end{align*}
Then we have the bound:
\[
\sup_{t \in (0, T^*)} N_{k,l}(t) \leq c_{k,l},
\]
for some constant $c_{k,l}$ depending on the initial data only through $N_{k,l}$.
\end{Theorem}
\begin{proof}
This time we induct on $l$. We have already established that the result holds with $l=0$ and for any $k\geq k_0$. Suppose that we have established the bound for $l-1$, where $l \geq 1$. We consider \eq{rhothetaeq22} and act on both sides with $\p_0^l$:
\[
- \Box \p_0^l \theta  = \ 2 \p_0^l(\p^\mu \rho \p_\mu \theta).
\]
Applying Lemma \ref{EnEst} on a time interval $(0, T)$ for $T<T^*$, we obtain:
\begin{align*}
\norm{\p_0^l \theta}{\dot{C}^{1}_t \dot{H}^k_x}^2 + \norm{\p_0^l \theta}{C^0_t \dot{H}^{k+1}_x }^2 &\leq 8 \norm{\p_0^l(\p^\mu \rho \p_\mu \theta)}{L^1_t \dot{H}^k_x}^2 \\& \qquad + 2\norm{(\p_0^l\theta)_{\{x^0=0\}}}{\dot{H}^{k+1}}^2 + 2\norm{(\p_0^{l+1}\theta)_{\{x^0=0\}}}{\dot{H}^{k}}^2.
\end{align*}
As previously, for $\abs{A} = k$ we can write:
\[
D^A \p_0^l \left(\p^\mu \rho \p_\mu \theta \right) = \sum_{\substack{B+C = A\\i+j=l}} c_{B, i} \p^\mu \p_0^i D^B \rho \p_\mu \p_0^j D^C  \theta
\]
This time we note that either $\abs{B} + i < k-\frac{n}{2}$ or $\abs{C} + j < k-\frac{n}{2}$, so that either
\[
\norm{\p^\mu \p_0^i D^B \rho (t, \cdot)}{C^0} \leq c N_{k, l-1}(t), \quad \textrm{or} \quad \norm{\p_\mu \p_0^j D^C \theta(t, \cdot)}{C^0} \leq c N_{k, l-1}(t)
\]
for some universal constant $c$.  Thus we can bound one of the factors of $\p^\mu \p_0^i D^B \rho \p_\mu \p_0^j D^C \theta$ in $C^0$ by $N_{k, l-1}(t)$ and the other we bound in $L^2$. Since $N_{k, l-1}\leq N_{k,l}$, by the induction assumption we see that we can bound $N_{k, l-1}(t)$ by some constant $c_{N_{k,l}}$ depending only on $N_{k,l}$. We thus have
\[
\norm{\p_0^l(\p^\mu \rho \p_\mu \theta)}{L^1_t \dot{H}^k_x} \leq c_{N_{k, l}} T \sup_{t \in (0, T)} N_{k,l}(t).
\]
We can again make a similar argument with \eq{rhothetaeq12}: act with $\p_0^l$ and then use Lemma \ref{EnEst}, controlling the nonlinear terms as above to obtain:
\[
\sup_{t \in (0, T)} N_{k,l}(t) \leq \sqrt{2}N_{k,l} + c_{N_{k}} T  \sup_{t \in (0, T)} N_{k,l}(t).
\]
The remainder of the proof follows precisely as the end of the proof of the previous Theorem.
\end{proof}
To complete the proof of Proposition \ref{Prop2}, we first note that if $ \theta_0, \theta_{1}$ and $M\rho_j, \rho_{j+1}$ for $0\leq j \leq l$ are smooth functions whose derivatives are uniformly bounded in $M$, then by making use of \eq{rhothetaeq22} and iteratively differentiating with respect to $\p_0$ we can show that $\theta_j$ is smooth with uniformly bounded derivatives for $j \leq l+2$. In particular $N_{k,l}$ is uniformly bounded for any $k$. Invoking Theorem \ref{tThm}, together with Morrey's inequality to control pointwise norms of the solution in terms of $N_{k,l}(t)$ we have established the result.

\subsection{Proof of Lemma \ref{ICLem}}
The full result of Lemma \ref{ICLem} follows from Theorem \ref{thm:EFTexp}, used iteratively with the equation for $\theta$ to remove all $x^0$ derivatives of $\theta$ of order greater than one. We choose instead to establish Lemma \ref{ICLem} for $l\leq 3$ explicitly, and leave the general case for the reader. We first need to find $\rho_l, \theta_l$ in terms of $\rho_0, \rho_1, \theta_0, \theta_1$ by using the equations satisfied by $\rho, \theta$. For instance, \eq{rhothetaeq1} may be rearranged to give:
\ben{rhott}
\p_0^2 \rho = \p_i\p_i \rho -M^2 \rho + \p_\mu \rho \p^\mu \rho - \p_\mu \theta \p^\mu \theta - M^2 \rho^2 W(\rho)
\een
which on restricting to $\{x^0=0\}$ gives:
\ben{rho2}
\rho_2 = \p_i\p_i \rho_0 -M^2 \rho_0 + \p_i \rho_0 \p_i \rho_0 - \rho_1^2 - \p_i \theta_0 \p_i \theta_0 + \theta_1^2 - M^2 \rho_0^2 W(\rho_0).
\een
Similarly, we find:
\ben{theta2}
\theta_2 =  \p_i\p_i \theta_0 -2 \rho_1 \theta_1 + 2 \p_i \rho_0 \p_i \theta_0.
\een
We can also compute higher derivatives. Differentiating \eq{rhott} with respect to $x^0$ gives:
\[
\p_0^3 \rho = \p_i\p_i \p_0\rho -M^2 \p_0\rho + 2 \p_\mu \p_0\rho \p^\mu \rho - 2 \p_\mu \p_0\theta \p^\mu \theta - M^2 \p_0\rho[ 2\rho W(\rho) + \rho^2 W'(\rho)]
\]
restricting to $\{x^0=0\}$ gives:
\begin{align*}
\rho_3 &=  \p_i\p_i \rho_1 - M^2 \rho_1 -2\rho_2\rho_1 + 2 \p_i \rho_1 \p_i \rho_0+2\theta_2\theta_1 - 2 \p_i \theta_1 \p_i \theta_0\\&\qquad- M^2 \rho_1[ 2\rho_0 W(\rho_0) + \rho_0^2 W'(\rho_0)]
\end{align*}
and similarly
\[
\theta_3 = \p_i \p_i \theta_1 - 2 \rho_2 \theta_1 - 2 \rho_1 \theta_2 + 2 \p_i \rho_1 \p_i \theta_0+ 2 \p_i \rho_0 \p_i \theta_1.
\]
Here, we understand $\rho_2, \theta_2$ to be given in terms of $\rho_0, \rho_1, \theta_0, \theta_1$ through \eq{rho2}, \eq{theta2}. We can continue in this fashion to iteratively construct $\rho_l, \theta_l$ for all $l$.

Now, let us consider when it will occur that $M\rho_0, M\rho_1, \rho_2, \theta_0, \theta_1, \theta_2$ and their derivatives are uniformly bounded in $M$. We see that \eq{rho2} constrains $\rho_0$, while \eq{theta2} imposes no constraint.  A necessary and sufficient condition is that
\[
\rho_0 = \frac{R_2}{M^2}, \qquad \rho_1 = \frac{\tilde{R}_1}{M},
\]
where here and below $R_j$ and $\tilde{R}_j$ are smooth remainder terms whose derivatives are uniformly bounded in $M$.

Next, we consider when $M\rho_0, M\rho_1, M\rho_2, \rho_3, \theta_0, \theta_1, \theta_2, \theta_3$ and derivatives will be bounded in  $M$. We see from examining  the expressions for $\rho_2, \rho_3$ that a necessary and sufficient condition  is
\[
\rho_0 = \frac{\theta_1^2-\p_i \theta_0 \p_i \theta_0}{M^2} + \frac{R_3}{M^3}, \qquad \rho_1 = \frac{\tilde{R}_2}{M^2}
\]
On the other hand, \eq{theta2} imposes no further constraint.

Next, we consider when $M\rho_j, \rho_{j+1}, \theta_j, \theta_{j+1}$ will be bounded for $0\leq j \leq 3$. We can find an expression for $\rho_4$, and although we do not give it explicitly, we can see that $\rho_4$ is bounded if and only if $M^2 \rho_2$ is bounded. Making use of this fact, and examining also  the expressions for $\rho_2, \rho_3$ and $\theta_2$ we see that a necessary and sufficient condition is given by:
\[
\rho_0 = \frac{\theta_1^2-\p_i \theta_0 \p_i \theta_0}{M^2} + \frac{R_4}{M^4}, \qquad \rho_1 = \frac{2 \theta_1 \p_i \p_i \theta_0 - 2\p_i \theta_0 \p_i \theta_1}{M^2}+ \frac{\tilde{R_3}}{M^3}
\]
We can continue iteratively in this way to establish the result for higher $l$.

\subsection{Approximation of UV solutions by EFT solutions}

In order to show that the EFT solutions indeed approximate the true solution, we shall require the following result:
\begin{Lemma}\label{dualen} 
Suppose that $u: \mathcal{M}_T \to \R$ is a smooth function satisfying:
\[
\begin{array}{c}
\displaystyle- \Box u =\p_\mu X^\mu+ f \\[.4cm]
u|_{x^0=0} = u_0, \qquad \p_0 u|_{x^0=0} = u_1.
\end{array}
\]
Then we have the bound:
\[
\norm{u}{C^0_tL^2_x} \leq c(1+T) \left(   \norm{X}{C^0_t L^2_x} +  \norm{f}{L^1_tL^2_x} +  \norm{u_1}{L^2} + \norm{u_0}{L^2}\right)
\]
for some constant $c$. Furthermore, if $u_0, u_1, X^\mu, f$ are $O_{\infty}(M^{-l})$, then $u = O_{\infty}(M^{-l})$.
\end{Lemma}
\begin{proof}
Let $\eta \in C^\infty_c(\mathcal{M}_{T})$, and let $\chi$ be given by \eq{chidef}. By \eq{WEcomp} we have:
\begin{align*}
\int_{\mathcal{M}_T} u \eta &= \int_{\mathcal{M}_T} \left( (\p_\mu X^\mu) \chi + f \chi\right) + \int_{\{x^0=0\}}  \left[\ u_1 \chi - u_0 (\p_0 \chi)  \right] \\
&=\int_{\mathcal{M}_T} \left(  -X^\mu \p_\mu \chi + f \chi\right) + \int_{\{x^0=0\}}  \left[\ (u_1+X^0) \chi - u_0 (\p_0 \chi)  \right]
\end{align*}
Noting that
\[
\frac{d}{dt} \int_{x^0 = t} \chi^2 = \int_{x^0 = t} 2 \chi \p_0 \chi \quad \implies \quad \norm{\chi}{C^0_tL^2_x} \leq T  \norm{\chi}{\dot{C}^1_tL^2_x}
\]
and invoking Lemma \ref{EnEst}, we deduce:
\[
\abs{\int_{\mathcal{M}_T} u \eta} \leq c(1+T) \left(   \norm{X}{C^0_t L^2_x} +  \norm{f}{L^1_tL^2_x} +  \norm{u_1}{L^2} + \norm{u_0}{L^2}\right)) \norm{\eta}{L^1_t L^2_x},
\]
for a constant $c$. By a density argument this holds for all $\eta \in L^1((0, T);  L^2(\T^n))$.

Now, let $m = \sup_{t \in (0, T)} \norm{u (\cdot, t)}{L^2}$. If $m=0$ there is nothing to show, otherwise for $0<\epsilon<m$ set:
\[
A_{\epsilon} := \{ t: \norm{u (\cdot, t)}{L^2}> m-\epsilon \}.
\]
By the continuity of $u$ we have $|A_\epsilon|>0$. Let $\eta(x,t) = \frac{\mathds{1}_{A_\epsilon}(t)}{|A_{\epsilon}|} \frac{u(x,t)}{\norm{u (\cdot, t)}{L^2}}$. Then:
\[
\int_{\mathcal{M}_{T}} u \eta = \frac{1}{|A_{\epsilon}|} \int_{A_\epsilon}  \left(\int_{\T^n}  u^2\right)^{\frac{1}{2}} >  m-\epsilon
\]
and
\[
\norm{\eta}{L^1_t L^2_x} = \int_{[0,T]}  \left(\frac{\mathds{1}_{A_\epsilon}}{|A_{\epsilon}|} \frac{1}{\norm{u (\cdot, t)}{L^2}} \left( \int_{\T^n}  u^2 \right)^{\frac{1}{2}}\right) =1
\]
Thus we have:
\[
m-\epsilon \leq c(1+T) \left(   \norm{X}{C^0_t L^2_x} +  \norm{f}{L^1_tL^2_x} +  \norm{u_1}{L^2} + \norm{u_0}{L^2}\right),
\]
since $\epsilon$ was arbitrary, we conclude that 
\[m=\norm{u}{C^0_t L^2_x} \leq c(1+T) \left(   \norm{X}{C^0_t L^2_x} +  \norm{f}{L^1_tL^2_x} +  \norm{u_1}{L^2} + \norm{u_0}{L^2}\right).
\]
The final part of the result follows by noting that for any multi-index $A$, $D^A u$ satisfies 
\[
\begin{array}{c}
\displaystyle- \Box D^Au =\p_\mu (D^AX^\mu)+ D^Af \\[.4cm]
D^Au|_{x^0=0} = D^Au_0, \qquad \p_0 D^Au|_{x^0=0} = D^Au_1,
\end{array}
\]
so applying the first part of the Lemma together with Morrey's inequality completes the proof.
\end{proof}

%
Before we prove the general case of Theorem  \ref{Thm:EFTApr}, we prove the special case of $l=2$ to indicate how the proof proceeds in the general case.
\begin{Lemma}
Suppose the conclusions of Theorem \ref{thm:EFTexp} hold with $l=2$. Let $\tilde{\theta}:\mathcal{M}_{T_*} \to \R$ be an EFT solution to order $2$.  Then for $j\leq 2$:
\[
\p_0^j \theta=\p_0^j \tilde{\theta} + O_{\infty}(M^{-3+j}).
\]
\end{Lemma}
\begin{proof}
We note that $\tilde{\theta}$ satisfies:
\[
\textstyle- \Box \tilde{\theta} = -\frac{2}{M^2}\p_\mu\left(\p_\nu \tilde{\theta} \p^\nu\tilde{\theta} \right)\p^\mu\tilde{\theta}   + \frac{R}{M^3}, \qquad \tilde{\theta}|_{\{x^0=0\}} = \theta_0 , \qquad \p_0\tilde{\theta}|_{\{x^0=0\}} = \theta_1,
\]
so that $\delta\theta = \theta - \tilde{\theta}$ satisfies:
\[
- \Box \delta \theta = 2\p_\mu \rho \p^\mu \theta + \frac{2}{M^2}\p_\mu\left(\p_\nu \tilde{\theta} \p^\nu\tilde{\theta} \right)\p^\mu\tilde{\theta}   - \frac{R}{M^3}
\]
with initial conditions:
\[
\delta{\theta}|_{\{x^0=0\}} = \p_0\delta{\theta}|_{\{x^0=0\}} =0.
\]

Our assumptions give us $\p_0^j \theta = O_{\infty}(1)$ for $0\leq j \leq 4$, and from the definition of an EFT solution we have a similar bound on $\tilde{\theta}$. Combining these facts with the bounds on $\rho$ given by Theorem \ref{thm:EFTexp}, we see that 
\[
- \Box \delta \theta = R'
\]
for some $R'$ satisfying $R' = O_{\infty}(M^{-2})$, $\p_0 R' = O_{\infty}(M^{-1})$. By Lemma \ref{EFTEnEst} we deduce that
\[
\p_0^2 \delta \theta = O_{\infty}(M^{-1}), \qquad \p_0 \delta \theta = O_{\infty}(M^{-2}), \qquad  \delta \theta = O_{\infty}(M^{-2}).
\]

Now, we rewrite the equation for  $\delta\theta$ as:
\begin{align*}
- \Box \delta \theta &= 2\p_\mu\left( \rho + \frac{\p_\nu {\theta} \p^\nu{\theta}}{M^2}  \right) \p^\mu \theta  - \frac{2}{M^2} \p_\mu \left( \p_\nu \delta \theta \p^\nu {\theta} +  \p_\nu  \tilde{\theta} \p^\nu \delta{\theta} \right) \p^\mu \theta \\
&\quad  -\frac{2}{M^2}\p_\mu\left(\p_\nu \tilde{\theta} \p^\nu\tilde{\theta} \right)\p^\mu\delta \theta
 - \frac{R}{M^3}
\end{align*}
with initial conditions:
\[
\delta{\theta}|_{\{x^0=0\}} = \p_0\delta{\theta}|_{\{x^0=0\}} =0.
\]
This is of the form
\[
-\Box \delta \theta = \p_\mu \tilde{X}^\mu+ f
\]
with:
\begin{align*}
\tilde{X}^\mu &= 2\left( \rho + \frac{2}{M^2} \p_\nu {\theta} \p^\nu{\theta} \right) \p^\mu \theta  - \frac{2}{M^2} \left( \p_\nu \delta \theta \p^\nu {\theta} +  \p_\nu  \tilde{\theta} \p^\nu \delta{\theta} \right) \p^\mu \theta 
\end{align*}
and
\begin{align*}
f &= -2\left( \rho + \frac{2}{M^2} \p_\nu {\theta} \p^\nu{\theta} \right) \Box \theta + \frac{2}{M^2} \left( \p_\nu \delta \theta \p^\nu {\theta} +  \p_\nu  \tilde{\theta} \p^\nu \delta{\theta} \right) \Box \theta\\
&\qquad  -\frac{2}{M^2}\p_\mu\left(\p_\nu \tilde{\theta} \p^\nu\tilde{\theta} \right)\p^\mu\delta \theta
 - \frac{R}{M^3}
\end{align*}
where now we have $f, \tilde{X}^\mu = O_\infty(M^{-3})$. Hence, applying Lemma \ref{dualen} completes the proof.
\end{proof}

We are now ready to prove the full result
\begin{T2}
Suppose that $(\rho, \theta)$ satisfy the conclusions of Theorem  \ref{thm:EFTexp} on $\mathcal{M}_T$ for some $l$, and suppose that $\tilde{\theta}$ is an EFT solution of \eq{EFT_exp} to order $l$ such that
\[
\theta|_{x^0=0}= \tilde{\theta}|_{x^0=0}, \qquad \p_0\theta|_{x^0=0}= \p_0\tilde{\theta}|_{x^0=0}.
\]
Then for any  $j\leq l$ we have
\ben{ErrorEst}
\p_0^{j}  \theta =  \p_0^j \tilde{\theta}  + O_{\infty}(M^{-l-1+j}).
\een
In particular, $\theta-\tilde{\theta} = O(M^{-l-1})$.
\end{T2}
\begin{proof}
We note that $\tilde{\theta}$ satisfies:
\[
\textstyle- \Box \tilde{\theta} =  2\p_\mu \left(\frac{\mathcal{F}_2[\tilde{\theta}]}{M^2}   +\ldots+  \frac{\mathcal{F}_{2\lfloor \frac{l}{2} \rfloor }[\tilde{\theta}]}{M^{2\lfloor \frac{l}{2} \rfloor }} \right) \p^\mu \tilde{\theta} + \frac{R}{M^{l+1}} , \qquad \tilde{\theta}|_{\{x^0=0\}} = \theta_0 , \qquad \p_0\tilde{\theta}|_{\{x^0=0\}} = \theta_1,
\]
for some remainder $R$ which is uniformly bounded along with all its derivatives. Accordingly, $\delta\theta = \theta - \tilde{\theta}$ satisfies:
\[
- \Box \delta \theta = R' :=  2\p_\mu \rho \p^\mu \theta -2 \p_\mu \left(\frac{\mathcal{F}_2[\tilde{\theta}]}{M^2}   +\ldots+  \frac{\mathcal{F}_{2\lfloor \frac{l}{2} \rfloor }[\tilde{\theta}]}{M^{2\lfloor \frac{l}{2} \rfloor }} \right) \p^\mu \tilde{\theta} - \frac{R}{M^{l+1}}
\]
with initial conditions:
\[
\delta{\theta}|_{\{x^0=0\}} = \p_0\delta{\theta}|_{\{x^0=0\}} =0.
\]
We first note that our assumptions imply $\p_0^j R' = O_\infty(M^{-2})$ for $j < l-1$ and $\p_0^{l-1} R' = O_\infty(M^{-1})$. By Lemma \ref{EFTEnEst}, we deduce that $\p_0^j \delta \theta = O_{\infty}(M^{-2})$ for $j \leq l-1$ and $\p_0^l \delta \theta = O_{\infty}(M^{-1})$. This establishes \eq{ErrorEst} for $j=l, l-1$. 

We now perform an induction backwards in $j$. Suppose that for some $J>1$ we have established that $\p_0^j \delta \theta = O_{\infty}(M^{-l-1+j})$ for $j> J$ and $\p_0^j \delta \theta = O_{\infty}(M^{-l-1+J})$ for $j \leq J$. We write the expression for $R'$ as:
\begin{align}
R' &= 2\p_\mu \left( \rho -  \frac{\mathcal{F}_2[{\theta}]}{M^2}   -\ldots-  \frac{\mathcal{F}_{2\lfloor \frac{l+1-J}{2} \rfloor }[{\theta}]}{M^{2\lfloor \frac{l+1-J}{2} \rfloor }} \right)\p^\mu \theta \nonumber \\
&\qquad + 2\p_\mu \left(  \frac{\mathcal{F}_2[{\theta}]}{M^2}   +\ldots+  \frac{\mathcal{F}_{2\lfloor \frac{l+1-J}{2} \rfloor }[{\theta}]}{M^{2\lfloor \frac{l+1-J}{2} \rfloor }} \right)\p^\mu \delta \theta 
\nonumber \\
&\qquad + 2\p_\mu \left(  \frac{\mathcal{F}_2[{\theta}]- \mathcal{F}_2[\tilde{\theta}]}{M^2}   +\ldots+  \frac{\mathcal{F}_{2\lfloor \frac{l+1-J}{2} \rfloor }[{\theta}]- \mathcal{F}_{2\lfloor \frac{l+1-J}{2} \rfloor }[\tilde{\theta}]}{M^{2\lfloor \frac{l+1-J}{2} \rfloor }} \right)\p^\mu \tilde{ \theta}
\nonumber\\&\qquad -2 \p_\mu \left(\frac{\mathcal{F}_{2\lfloor \frac{l+1-J}{2} \rfloor +2}[\tilde{\theta}]}{M^{2\lfloor \frac{l+1-J}{2} \rfloor +2}}   +\ldots+  \frac{\mathcal{F}_{2\lfloor \frac{l}{2} \rfloor }[\tilde{\theta}]}{M^{2\lfloor \frac{l}{2} \rfloor }} \right) \p^\mu \tilde{\theta} - \frac{2R}{M^{l+1}} \nonumber
\end{align}
We split $R'$ into the terms on each of the four lines above as $R' = R_1' + R_2' + R_3' + R_4'$. We claim that $\p_0^j R' = O_\infty(M^{-l-2+J})$ for $j \leq J-2$. If this is true, then by Lemma \ref{EFTEnEst} we deduce that $\p_0^j \delta \theta = O_{\infty}(M^{-l-1+j})$ for $j> J-1$ and $\p_0^j \delta \theta = O_{\infty}(M^{-l-2+J})$ for $j \leq J-1$, so by induction we will have established \eq{ErrorEst} for $1\leq j \leq l$ and $\delta \theta = O_\infty(M^{-l})$.

To establish the claim, we note that Theorem \ref{thm:EFTexp} immediately establishes that $\p_0^j R'_1= O_\infty(M^{-l-2+J})$ for $j \leq J-2$. For the second term, we note that by our induction assumption $\p_0^j \p_\mu \delta \theta = O_{\infty}(M^{-l-1+J})$ for $j\leq J-2$. Moreover, $\p_0^j \p_\mu \mathcal{F}_{2k}[\theta]$ contains at worst $\p_0^{2k+j} \theta$, so since $2\lfloor \frac{l+1-J}{2} \rfloor +J-2 \leq l$ we can bound the first factor by $M^{-2}$, and hence we have $\p_0^j R'_2= O_\infty(M^{-l-2+J})$ for $j \leq J-2$. To bound the third term, we note that we can factor
\[
\mathcal{F}_{2k}[\theta] - \mathcal{F}_{2k}[\tilde{\theta}] = \sum_{j\leq 2k-1} (\p_0^j \delta \theta )P_{jk}[\theta, \tilde{\theta}]
\]
where $P_{jk}[\theta, \tilde{\theta}]$ is a polynomial in $\theta, \tilde{\theta}$ and their derivatives up to order $2k-1$. For $j \leq J-2$ and $J-2 + 2k \leq l$, by the induction assumption $\p_0^j(\mathcal{F}_{2k}[\theta] - \mathcal{F}_{2k}[\tilde{\theta}] ) = O_{\infty}(M^{-l-3+J+2k})$. For $k \leq 2\lfloor \frac{l+1-J}{2}\rfloor$, we have  $J-2 + 2k \leq l$, so we conclude that $\p_0^j R'_3= O_\infty(M^{-l-2+J})$ for $j \leq J-2$. Finally, we manifestly have that $\p_0^j R'_4= O_\infty(M^{-l-2+J})$ for $j \leq J-2$ by the definition of the EFT solution.

To complete the proof, we note that we may write:
\begin{align}
R' &= 2\p_\mu \left( \rho -  \frac{\mathcal{F}_2[{\theta}]}{M^2}   -\ldots-  \frac{\mathcal{F}_{2\lfloor \frac{l}{2} \rfloor }[{\theta}]}{M^{2\lfloor \frac{l}{2} \rfloor }} \right)\p^\mu \theta \nonumber \\
&\qquad + 2\p_\mu \left(  \frac{\mathcal{F}_2[{\theta}]}{M^2}   +\ldots+  \frac{\mathcal{F}_{2\lfloor \frac{l}{2} \rfloor }[{\theta}]}{M^{2\lfloor \frac{l}{2} \rfloor }} \right)\p^\mu \delta \theta 
\nonumber \\
&\qquad + 2\p_\mu \left(  \frac{\mathcal{F}_2[{\theta}]- \mathcal{F}_2[\tilde{\theta}]}{M^2}   +\ldots+  \frac{\mathcal{F}_{2\lfloor \frac{l}{2} \rfloor }[{\theta}]- \mathcal{F}_{2\lfloor \frac{l}{2} \rfloor }[\tilde{\theta}]}{M^{2\lfloor \frac{l}{2} \rfloor }} \right)\p^\mu \tilde{ \theta}
\nonumber- \frac{2R}{M^{l+1}} \nonumber \\
&= \p_\mu X^\mu + \tilde{R} \nonumber
\end{align}
where
\begin{align*}
X^\mu &= 2\left( \rho -  \frac{\mathcal{F}_2[{\theta}]}{M^2}   -\ldots-  \frac{\mathcal{F}_{2\lfloor \frac{l}{2} \rfloor }[{\theta}]}{M^{2\lfloor \frac{l}{2} \rfloor }} \right)\p^\mu \theta \\
& \qquad + 2 \left(  \frac{\mathcal{F}_2[{\theta}]- \mathcal{F}_2[\tilde{\theta}]}{M^2}   +\ldots+  \frac{\mathcal{F}_{2\lfloor \frac{l}{2} \rfloor }[{\theta}]- \mathcal{F}_{2\lfloor \frac{l}{2} \rfloor }[\tilde{\theta}]}{M^{2\lfloor \frac{l}{2} \rfloor }} \right)\p^\mu \tilde{ \theta}
\end{align*}
and
\begin{align*}
\tilde{R} &= -2 \left( \rho -  \frac{\mathcal{F}_2[{\theta}]}{M^2}   -\ldots-  \frac{\mathcal{F}_{2\lfloor \frac{l}{2} \rfloor }[{\theta}]}{M^{2\lfloor \frac{l}{2} \rfloor }} \right)\Box \theta \nonumber \\
&\qquad + 2\p_\mu \left(  \frac{\mathcal{F}_2[{\theta}]}{M^2}   +\ldots+  \frac{\mathcal{F}_{2\lfloor \frac{l}{2} \rfloor }[{\theta}]}{M^{2\lfloor \frac{l}{2} \rfloor }} \right)\p^\mu \delta \theta 
\\
&\qquad - 2 \left(  \frac{\mathcal{F}_2[{\theta}]- \mathcal{F}_2[\tilde{\theta}]}{M^2}   +\ldots+  \frac{\mathcal{F}_{2\lfloor \frac{l}{2} \rfloor }[{\theta}]- \mathcal{F}_{2\lfloor \frac{l}{2} \rfloor }[\tilde{\theta}]}{M^{2\lfloor \frac{l}{2} \rfloor }} \right)\Box \tilde{ \theta}
\nonumber- \frac{2R}{M^{l+1}} \nonumber 
\end{align*}
By arguing as we did above when bounding $R_3'$, we can see that $f, X^\mu = O_{\infty}(M^{-l-1})$, so applying Lemma \ref{dualen} completes the proof.
\end{proof}

\section{Long timescales}\label{LTS}

In this section we will work on a time interval $[0, T)$, where $T\leq C M^\lambda$ for some  $0\leq \lambda <2$.  As we stated above, we do not assert that in general a solution to the UV theory will exist on such a time interval, but it may be that for particular choices of initial data a long-time solution may exist.  We wish to investigate the effectiveness of the EFT in approximating the UV solution on such a time interval. 

For simplicity, we work at an order involving only the first term in the EFT expansion, which has the advantage that truncating at this order gives rise to a quasilinear wave equation which is locally well posed. We \emph{assume} that a solution $(\rho, \theta)$ of the UV equation of motion exists on the interval $[0,T)$, satisfying the conditions of Theorem \ref{thm:EFTexp} to order 3. We recall that an EFT solution to order 3 is a function $\tilde{\theta}$ which is bounded, together with all its derivatives as $M \to \infty$ and satisfies:
\begin{align*}
&- \Box \tilde{\theta}  =  \frac{2}{M^2} \p^\mu (- \p_\mu \tilde{\theta}  \cdot \p^\mu \tilde{\theta} ) \p_\mu \tilde{\theta} + \frac{R_3}{M^4} \\
&\tilde{\theta}|_{\{x^0=0\}} = \theta|_{\{x^0=0\}}:= \theta_0  , \qquad \p_0 \tilde{\theta}|_{\{x^0=0\}} =\p_0\theta|_{\{x^0=0\}} :=\theta_1,
\end{align*}
for some $R_3$ which is bounded, together with all its derivatives as $M \to \infty$. 

We first assert that such a $\tilde{\theta}$ always exists. Returning to the proof of existence for EFT solutions in \S\ref{pf:EFT_sol}, we see that writing $\tilde{\theta}$ as a formal expansion in $1/M^2$ requires us to solve an iterated system of linear equations, each on the time interval $[0,T)$. It is well known that over long timescales secular growth may cause small perturbations to the source terms of a linear equation to have higher order effects in the solution. At $O(1)$ in the expansion we don't see secular growth, but the term of $O(M^{-2})$ would solve a linear equation whose right-hand side is at best bounded in time and so we'd expect the $O(M^{-2})$ term to grow like $T$. At the next order, the $O(M^{-4})$ term will be sourced by a term potentially growing in time like $T$, so we would expect the $O(M^{-4})$ term to grow like $T^2$. Continuing, we would expect the $O(M^{-2k})$ term to grow like $T^k$. To account for this, rather than expanding in $1/M^2$, we will instead expand in $T/M^2$. This will have the advantage that the terms in our expansion are uniformly bounded in $T$, since the secular growth is accounted for in the expansion coefficient.

We thus make a formal expansion of $\tilde{\theta}$ as:
\[
\tilde{\theta} = \tilde{\theta}^0 + \frac{T}{M^2} \tilde{\theta}^2 + \frac{T^2}{M^4} \tilde{\theta}^4 + \ldots + \frac{T^m}{M^{2m}} \tilde{\theta}^{2 m},
\]
where $m = \left \lceil \frac{2}{2-\lambda} \right \rceil$. Inserting this expansion into 
\[
- \Box \tilde{\theta} = \frac{2}{M^2} \p^\mu (- \p_\mu \tilde{\theta}  \cdot \p^\mu \tilde{\theta} ) \p_\mu \tilde{\theta} + \frac{T^m}{M^{2m+2}}R,
\]
we can equate powers of $M$ to obtain a sequence of linear equations:
\[
\begin{array}{ccccc}
- \Box \tilde{\theta}^0 = 0,& \quad &\tilde{\theta}^0|_{\{x^0=0\}} = \theta_0,& \quad &\p_0\tilde{\theta}^0|_{\{x^0=0\}} = \theta_1, \\
- T \Box \tilde{\theta}^2 = -2\p_\mu (\p_\nu \tilde{\theta}^0\p^\nu \tilde{\theta}^0) \p^\mu \tilde{\theta}^0, & \quad &\tilde{\theta}^2|_{\{x^0=0\}} = 0,& \quad &\p_0\tilde{\theta}^2|_{\{x^0=0\}} = 0, \\
\vdots&&\vdots &&\vdots \\
-T \Box \tilde{\theta}^{2m} = G_{2m}[\tilde{\theta}^0, \ldots, \tilde{\theta}^{2m-2}; T] & \quad &\tilde{\theta}^{2m}|_{\{x^0=0\}} = 0,& \quad &\p_0\tilde{\theta}^{2m}|_{\{x^0=0\}} = 0,
 \\
\end{array}
\]
and
\[
R = G_{2m+2}[\tilde{\theta}^0, \ldots, \tilde{\theta}^{2m}; T]
\]
Here, $G_{2j}[\tilde{\theta}^0, \ldots, \tilde{\theta}^{2j-2}; T]$ is a polynomial in derivatives of $\tilde{\theta}^0, \ldots, \tilde{\theta}^{2j-2}$, with coeffcients which are bounded uniformly in $T$. We can solve these equations sequentially on $\mathcal{M}_{T_*}$. At each stage we are solving a linear equation with smooth initial data and a smooth right-hand side, which is independent of $M$ and bounded in $T$. By commuting the equation with derivatives and applying Lemma \ref{EnEst}, we can bound $\tilde{\theta}^{2j}$ and its derivatives uniformly in $T$ and $M$, recalling that:
\[
\norm{f}{L^1_t \dot{H}^k_x} \leq T \norm{f}{C^0_t \dot{H}^k_x}.
\]
We finally conclude that we can bound $R$ and all its derivatives uniformly. Since $m (2- \lambda) \geq 2$, we see that $R_3 = M^{2-2m} T^{m}R$ satisfies the necessary boundedness conditions. 

Notice that if $\lambda$ is close to $2$ it is necessary to consider an expansion with $m \gg 3$ terms in order to construct a solution to the EFT to order $3$. We note that despite our expansion apparently including terms which are $O(M^{-2m})$, these terms do not `see' any higher order corrections to the EFT expansion. If we were to include the next correction in the EFT expansion in our equation, at highest order it would modify the equation for $-T \Box \tilde{\theta}^4$ by the addition of a term of $O(T^{-1})$, so even accounting for secular growth the effects of the $O(M^{-4})$ term in the EFT expansion would remain subleading.

In practice, it may be that there are more efficient ways to construct an EFT solution. For example, we could consider solving the truncated EFT equation 
\begin{align*}
&- \Box \tilde{\theta}  =  \frac{2}{M^2} \p^\mu (- \p_\mu \tilde{\theta}  \cdot \p^\mu \tilde{\theta} ) \p_\mu \tilde{\theta}  \\
&\tilde{\theta}|_{\{x^0=0\}}= \theta_1 , \qquad  \p_0 \tilde{\theta}|_{\{x^0=0\}} = \theta_1.
\end{align*}
i.e. equation \eqref{nonlinear_wave}. If this results in a solution $\tilde{\theta}$ which is bounded on $[0,T)$, together with its derivatives as $M \to \infty$, then $\tilde{\theta}$ will be an EFT solution to order $3$ (with $R_3 = 0$). A priori, however, it is not clear without further investigation whether a particular choice of initial data will lead to a uniformly bounded solution on such a time interval.

We conclude our discussion of long time intervals with the proof of Lemma \ref{LongTLemma}, which establishes that an EFT solution will approximate the UV solution, but only for a restricted range of $\lambda$. 
\begin{proof}[Proof of Lemma \ref{LongTLemma}]
We note that $\tilde{\theta}$ satisfies:
\[
\textstyle- \Box \tilde{\theta} = -\frac{2}{M^2}\p_\mu\left(\p_\nu \tilde{\theta} \p^\nu\tilde{\theta} \right)\p^\mu\tilde{\theta}   + \frac{R_3}{M^4}, \qquad \tilde{\theta}|_{\{x^0=0\}} = \theta_0 , \qquad \p_0\tilde{\theta}|_{\{x^0=0\}} = \theta_1,
\]
so that $\delta\theta = \theta - \tilde{\theta}$ satisfies:
\begin{align*}
- \Box \delta \theta &= 2\p_\mu\left( \rho + \frac{1}{M^2} \p_\nu {\theta} \p^\nu{\theta} \right) \p^\mu \theta  - \frac{2}{M^2} \p_\mu \left( \p_\nu \delta \theta \p^\nu {\theta} +  \p_\nu  \tilde{\theta} \p^\nu \delta{\theta} \right) \p^\mu \theta \\
&\quad  -\frac{2}{M^2}\p_\mu\left(\p_\nu \tilde{\theta} \p^\nu\tilde{\theta} \right)\p^\mu\delta \theta
 - \frac{R_3}{M^4}
\end{align*}
with initial conditions:
\[
\delta{\theta}|_{\{x^0=0\}} = \p_0\delta{\theta}|_{\{x^0=0\}} =0.
\]
The bounds assumed on the UV solution give us uniform control over $\p_\mu\p_\nu \theta$, and from the definition of an EFT solution we have a similar bound $\p_\mu\p_\nu \tilde{\theta}$. Combining these facts with the expansion for $\rho$ in the conclusion of Theorem \ref{thm:EFTexp}, we see
\[
- \Box \delta \theta = \frac{R'}{M^2}
\]
for some $R'$ satisfying $\sup_{\mathcal{M}_{T}} \abs{R'} \leq C$ where $C$ is independent of $M$ and may increase from line to line. This implies
\[
\norm{R'}{L^1_tL^2_x} \leq CT,
\]
so by Lemma \ref{EnEst} we deduce:
\[
\norm{\p_\mu \delta \theta}{C^0_tL^2_x} \leq C \frac{T}{M^2}.
\]

Returning to the original equation for $-\Box \delta \theta$, we observe that it may be written
\[
-\Box \delta \theta = \p_\mu \tilde{X}^\mu+ f
\]
with:
\begin{align*}
\tilde{X}^\mu &= 2\left( \rho + \frac{1}{M^2} \p_\nu {\theta} \p^\nu{\theta} \right) \p^\mu \theta  - \frac{2}{M^2} \left( \p_\nu \delta \theta \p^\nu {\theta} +  \p_\nu  \tilde{\theta} \p^\nu \delta{\theta} \right) \p^\mu \theta 
\end{align*}
and
\begin{align*}
f &= -2\left( \rho + \frac{1}{M^2} \p_\nu {\theta} \p^\nu{\theta} \right) \Box \theta + \frac{2}{M^2} \left( \p_\nu \delta \theta \p^\nu {\theta} +  \p_\nu  \tilde{\theta} \p^\nu \delta{\theta} \right) \Box \theta\\
&\qquad  -\frac{2}{M^2}\p_\mu\left(\p_\nu \tilde{\theta} \p^\nu\tilde{\theta} \right)\p^\mu\delta \theta
 - \frac{R_3}{M^4}
\end{align*}

where
\[
\norm{f}{L^1_tL^2_x}+  \norm{\tilde{X}}{C^0_tL^2_x} \leq \frac{C}{M^4}T^2.
\]
Hence, applying Lemma \ref{dualen} we deduce that
\[
\norm{\theta - \tilde{\theta}}{C^0_t L^2_x} \leq \frac{c T^3}{M^4}
\]
Commuting with spatial derivatives, and recalling that this does not affect behaviour in $M$, we finally conclude
\[
\theta= \tilde{\theta} + O_\infty \left( \frac{T^3}{M^4} \right).\qedhere
\]
\end{proof}

We remark that for $0 \leq \lambda < 2$ the method of expansion in $T/M^2$ will clearly permit us to construct EFT solutions to order $l$, provided that we expand in $T/M^2$ to sufficiently high order $m$. We claim (without proof) that if the UV solution is moreover assumed to satisfy the conclusions of Theorem \ref{thm:EFTexp} to order $l$, we can improve the estimate above to give
\[
 \norm{\theta-\tilde{\theta}}{C^0_t L^2_x}  \leq \frac{cT^{\lfloor \frac{l}{2}\rfloor+2}}{M^{l+1}},
\]
so that for any $0\leq \lambda < 2$  the EFT will approximate the UV theory in $\mathcal{M}_T$ for $l$ sufficiently large. This implies a modified uniqueness statement for EFT solutions, where two EFT solutions to order $l$ differ by $O\left( \frac{T^{\lfloor \frac{l}{2}\rfloor+1}}{M^{l+1}} \right)$.

\section{EFT and averaging} \label{Average}
\subsection{Averaging for linear equations}
Before we consider the proof of Theorem \ref{Prop3}, we first show how averaging methods can be applied to approximate solutions of linear equations in two situations relevant to our problem -- that of a Klein--Gordon equation whose mass is becoming large, and that of a wave equation with an oscillatory forcing term.
\subsubsection{The Klein--Gordon equation in the large mass limit} Let $\varrho : \mathcal{M}_T \to \R$ solve:
\[
\begin{array}{c}
\displaystyle- \Box \varrho + M^2 \varrho = F \\[.4cm]
\varrho|_{x^0=0} = \varrho_0, \qquad \p_0 \varrho|_{x^0=0} = \varrho_1
\end{array}
\]
for $\varrho_0, \varrho_1 \in C^\infty(\T^n)$, $F \in C^\infty(\mathcal{M}_T)$ satisfying $M\varrho_0, \varrho_1, F = O_{\infty}(1)$ as $M \to \infty$. Under these assumptions, by Lemma \eq{EnEst} we have that $\p_0 \varrho, M\varrho= O_{\infty}(1)$. In general, we do not expect $\p_0^2 \varrho$ to be bounded as $M \to \infty$. Indeed, $\rho(x^\mu) = M^{-1} \sin(M x^0)$ solves the equation with $F = 0$, $\varrho_0 = 0$, $\varrho_1 = 1$, and for this solution we clearly have $\p_0^k \varrho \sim M^{k-1}$.

We wish to investigate $\int_{\mathcal{M}_T} \varrho\eta$ in the limit $M \to \infty$ for some $\eta \in C^\infty_c(\mathcal{M})$ which is independent of $M$. To do this, we can make use of the equation in the form:
\[
\varrho = \frac{\Box \varrho}{M^2} + \frac{F}{M^2}
\]
and integrate by parts, discarding boundary terms since $\eta$ is compactly supported, to find:
\begin{align*}
\int_{\mathcal{M}_T} \varrho\eta &= \int_{\mathcal{M}_T} \left( \frac{\Box \varrho}{M^2} + \frac{F}{M^2} \right) \eta \\
&= \frac{1}{M^2} \int_{\mathcal{M}_T} F \eta + \frac{1}{M^2} \int_{\mathcal{M}_T}  \varrho \Box\eta
\end{align*}
Now, noting that $\Box \eta \in C^\infty_c(\mathcal{M})$, we can recursively rewrite the second term and find:
\begin{align*}
\int_{\mathcal{M}_T} \varrho\eta &= \frac{1}{M^2} \int_{\mathcal{M}_T} F \eta + \frac{1}{M^4} \int_{\mathcal{M}_T}  F \Box\eta + \frac{1}{M^4} \int_{\mathcal{M}_T}  \varrho \Box^2\eta \\
&= \int_{\mathcal{M}_T} F \left( \frac{\eta}{M^2} + \frac{\Box \eta}{M^4}+\cdots + \frac{\Box^{k-1}\eta}{M^{2k}}\right) + \frac{1}{M^{2k}} \int_{\mathcal{M}_T}  \varrho \Box^{k}\eta
\end{align*}
We know that $ \rho = O_\infty(M^{-1})$, and since $\eta$ is smooth, compactly supported, and independent of $M$ we deduce:
\[
\int_{\mathcal{M}_T}  \varrho \Box^{k}\eta = O(M^{-1}).
\]
Integrating by parts we conclude:
\[
\int_{\mathcal{M}_T} \varrho\eta = \int_{\mathcal{M}_T} \left(\frac{F}{M^2} + \frac{\Box F}{M^4}+\cdots + \frac{\Box^{k-1}F}{M^{2k}}  \right) \eta + O(M^{-2k-1}).
\]
We understand this equation as telling us that any average of $\rho$, weighted by a smooth, compactly supported, function $\eta$ will coincide with the corresponding average of the function $M^{-2} F + M^{-4} \Box F+\cdots + M^{-2k}\Box^{k-1}F$ up to corrections of order $M^{-2k-1}$. It will certainly not be the case that these functions agree pointwise to this order -- the example above makes this clear.

If we introduce the space of distributions $\mathscr{D}'(\mathcal{M}_T)$, with its usual topology, we say that a family of distributions $(u_M)_{M\geq M_0}$ is $O_{\mathscr{D}'}(M^{-k})$ as $M \to \infty$ if the set $\{M^k u_M: {M\geq M_0}\}$ is bounded in $\mathscr{D}'(\mathcal{M}_T)$. What we have shown is (with the canonical identification of locally integrable functions as distributions):
\[
\varrho = \frac{F}{M^2} + \frac{\Box F}{M^4}+\cdots + \frac{\Box^{k-1}F}{M^{2k}} + O_{\mathscr{D}'}(M^{-2k-1})
\]
holds for any $k =1, 2, \ldots$. Thus we have found an asymptotic series expansion of $\varrho$ which holds in the sense of distributions, i.e.\ in the sense of local weighted averages.

\subsubsection{A wave equation with highly oscillatory source term} Let $\vartheta : \mathcal{M}_T \to \R$ solve:
\[
\begin{array}{c}
\displaystyle- \Box \vartheta = F \\[.4cm]
\vartheta|_{x^0=0} = \vartheta_0, \qquad \p_0 \vartheta|_{x^0=0} = \vartheta_1
\end{array}
\]
for $\vartheta_0, \vartheta_1 \in C^\infty(\T^n)$, $F \in C^\infty(\mathcal{M}_T)$. We again wish to investigate $\int_{\mathcal{M}_T} \vartheta\eta$ for $\eta \in C^\infty_c(\mathcal{M}_T)$. To do this, we introduce the auxiliary function $\chi$ which solves:
\ben{chidef}
\begin{array}{c}
\displaystyle- \Box \chi = \eta \\[.4cm]
\chi|_{x^0=T} = 0, \qquad \p_0 \chi|_{x^0=T} = 0.
\end{array}
\een
We note for later the standard energy estimate:
\ben{Eest}
\norm{\p \chi}{C^0_t L^2_x} \leq  c\norm{\eta}{L^1_t L^2_x}.
\een

We now compute:
\begin{align}
\int_{\mathcal{M}_T} \vartheta\eta &= -\int_{\mathcal{M}_T} \vartheta\Box \chi = -\int_{\mathcal{M}_T} \Box  \vartheta\chi + \int_{\{x^0=0\}} \left[\left(\p_0 \vartheta\right) \chi - \vartheta \left(\p_0 \chi\right)  \right] \nonumber \\
&= \int_{\mathcal{M}_T} F \chi + \int_{\{x^0=0\}} \left[\ \vartheta_1 \chi - \vartheta_0 (\p_0 \chi)  \right]. \label{WEcomp}
\end{align}
Here we have integrated by parts, using the conditions on $\chi$ at $x^0=T$ to discard future boundary terms.

Thus we have related the average of $\vartheta$, weighted against $\eta$, to the data for $\vartheta$ together with quantities computable from $\eta$. By itself, this computation does not gain us a great deal. Let us suppose, however, that $F$ is highly oscillatory in time, say:
\[
F = \kappa(x^i) \cos (M x^0),
\]
where we understand $M$ to be large and assume $\kappa$ is smooth and independent of $M$. We can integrate by parts to obtain
\begin{align*}
 \int_{\mathcal{M}_T} F \chi &=  \int_{\mathcal{M}_T} \kappa \p_0 \left(\frac{ \sin (Mx^0)}{M}\right) \chi = -\frac{1}{M} \int_{\mathcal{M}_T} \kappa \sin(M x^0) \p_0 \chi  \\
 &= \frac{1}{M} \int_{\mathcal{M}_T} \kappa\p_0 \left(\frac{ \cos (Mx^0)}{M}\right) \p_0 \chi \\
 &= - \frac{1}{M^2} \int_{\{x^0 = 0\}} \kappa \p_0 \chi -\frac{1}{M^2} \int_{\mathcal{M}_T} \kappa \cos (Mx^0) \p_0\p_0 \chi  \\
 &= - \frac{1}{M^2} \int_{\{x^0 = 0\}} \kappa \p_0 \chi -\frac{1}{M^2} \int_{\mathcal{M}_T} \kappa \cos (Mx^0) (\p_i\p_i \chi+\eta) \\
 &=   - \frac{1}{M^2} \int_{\{x^0 = 0\}} \kappa \p_0 \chi -\frac{1}{M^2} \int_{\mathcal{M}_T} \p_i \p_i\kappa \cos (Mx^0) \chi- \frac{1}{M^2} \int_{\mathcal{M}_T} \kappa \cos (Mx^0)\eta.
\end{align*}
In the second-to-last line we have made use of the equation satisfied by $\chi$. We now note two things. Firstly, since $\eta$ is compactly supported, we can continue to integrate by parts without picking up boundary terms to see
\[
\int_{\mathcal{M}_T} \kappa \cos (Mx^0)\eta = O(M^{-k})
\]
for any $k\in \N$. Secondly, we can repeat the first integration by parts of our previous computation to see that:
\begin{align*}
-\frac{1}{M^2} \int_{\mathcal{M}_T} \p_i \p_i\kappa \cos (Mx^0) \chi&=  \frac{1}{M^3} \int_{\mathcal{M}_T} \p_i \p_i\kappa \sin (Mx^0) \p_0 \chi = O(M^{-3}).
\end{align*}
Returning to \eq{WEcomp}, we have established:
\[
\int_{\mathcal{M}_T} \vartheta\eta = \int_{\{x^0=0\}} \left[\ \vartheta_1 \chi - \left(\vartheta_0 + \frac{\kappa}{M^2}\right) (\p_0 \chi)  \right]+O(M^{-3})
\]
Now in principle we have found an asymptotic expansion for $\int_{\mathcal{M}_T} \vartheta\eta$. By repeating the same arguments, we find
\[
\int_{\mathcal{M}_T} \vartheta\eta = \int_{\{x^0=0\}} \left[\ \vartheta_1 \chi - \left(\vartheta_0 + \frac{\kappa}{M^2}-  \frac{\p_i \p_i \kappa}{M^4}\right) (\p_0 \chi)  \right]+O(M^{-5})
\]
and so on. In this form, however, it's difficult to interpret the meaning of this expansion. To see this more plainly, we return to \eq{WEcomp} and interpret the identity in the opposite direction. That is, \eq{WEcomp} permits us to recast certain terms involving integration over $\{x^0 =0\}$ as integration against $\eta$ of a particular solution to the wave equation. More concretely, let $\tilde{\vartheta}: \mathcal{M}_T \to \R$ solve
\[
\begin{array}{c}
\displaystyle- \Box \tilde{\vartheta} = 0 \\[.4cm]
\tilde{\vartheta}|_{x^0=0} = \vartheta_0 + \frac{\kappa}{M^2}-  \frac{\p_i \p_i \kappa}{M^4}, \, \qquad \p_0 \tilde{\vartheta}|_{x^0=0} = \vartheta_1.
\end{array}
\]
We have shown that:
\[
\int_{\mathcal{M}_T} (\vartheta - \tilde{\vartheta}) \eta = O(M^{-5}).
\]
or equivalently:
\[
\vartheta = \tilde{\vartheta} +O_{\mathscr{D}'}(M^{-5})
\]
By continuing with our expansion, we see that we can approximate the solution of our original equation to arbitrary accuracy by a solution of the \emph{free} wave equation with corrected initial data.

\subsection{Proof of Theorem \ref{Prop3}}\label{AvgProof}

We now turn to our original problem and attempt to apply the same arguments. We are presented with two complications. Firstly, we do not have an explicit expression for $\rho$, and so we have to rely on the equation it satisfies when integrating by parts. Secondly, the equation is nonlinear and so the computations we require will be more complicated.

We assume that initial data for $\theta, \rho$ are given satisfying the conditions of Proposition \ref{Prop1} and take $T$ to be the uniform existence time established in that result. Throughout we will use freely the bounds on $\rho, \theta$ given in Proposition \ref{Prop1} in order to estimate terms. In particular, we can bound $\theta, \p_0 \theta$, $\p_0^2 \theta$, $M \rho$ and $\p_0 \rho$ uniformly in $M$, together with as many spatial derivatives of these quantities as we desire. We shall also require the following result, which loosely says that the local energy density in the $\rho$ field is constant in time:
\begin{Lemma}\label{EngLem}
Let
\[
 \varepsilon = \frac{1}{2} \left( \rho_1^2 + M^2 \rho_0^2\right),
\]
be the initial energy density of the $\rho$ field in the limit $M \to \infty$. We have:
\[
\norm{(\p_0 \rho)^2  + M^2 \rho^2 - 2\varepsilon}{L^\infty(\mathcal{M}_T)} \leq \frac{C}{M} 
\]
\end{Lemma}
\begin{proof}
Let 
\[
\mathcal{E} = \frac{1}{2} (\p_0 \rho)^2 + \frac{1}{2}(M \rho +M^{-1}\p_\nu \theta \p^\nu \theta)^2 + M^2 \rho^3  + (\p_0 \rho)^2 \rho.
\]
Using the equation for $\rho$, we compute:
\begin{align*}
\p_0 \mathcal{E} &= (\p_0 \rho+2\rho \p_0 \rho)  \p_0^2 \rho+  (M \p_0 \rho + 2 M^{-1}\p_\nu \p_0 \theta \p^\nu \theta)(M \rho + M^{-1}\p_\nu \theta \p^\nu \theta)\\&\qquad  + (\p_0 \rho)^3  + 3 M^2 \rho^2 \p_0 \rho \\
&= (\p_0 \rho+2\rho \p_0 \rho) \left(-M^2 \rho -(\p_0 \rho)^2 - \p_\nu \theta \p^\nu \theta - M^2 \rho^2  \right) \\
&\qquad+ (\p_0 \rho)^3 + M \p_0 \rho (M \rho + M^{-1}\p_\nu \theta \p^\nu \theta) + 3 M^2 \rho^2 \p_0 \rho \\&\qquad  + 2 M^{-1}\p_\nu \p_0 \theta \p^\nu \theta(M \rho + M^{-1}\p_\nu \theta \p^\nu \theta)\\
&\qquad + (\p_0 \rho+2\rho \p_0 \rho) \left( \p_i\p_i \rho + \p_i \rho\p_i \rho- M^2 \rho^3 \frac{e^{2\rho} - 1-2\rho - 2\rho^2}{2\rho^3} \right) \\
\end{align*}
We see that all of the terms which are not at least $O(M^{-1})$ cancel to give:
\begin{align*}
\p_0 \mathcal{E} & = -2 \rho(\p_0 \rho)^3 -2 \rho \p_0 \rho \p_\nu \theta \p^\nu \theta -2M^2 \rho^3\p_0 \rho \\
&\qquad  + 2 M^{-1}\p_\nu \p_0 \theta \p^\nu \theta(M \rho + M^{-1}\p_\nu \theta \p^\nu \theta)\\
&\qquad + (\p_0 \rho+2\rho \p_0 \rho) \left( \p_i\p_i \rho + \p_i \rho\p_i \rho- M^2 \rho^3 \frac{e^{2\rho} - 1-2\rho - 2\rho^2}{2\rho^3} \right).
\end{align*}
So that:
\[
\norm{\p_0 \mathcal{E}}{L^\infty(\mathcal{M}_T)} \leq \frac{C}{M}
\]
Integrating in time, we deduce that 
\[
\sup_{x^i \in \T^n} \abs{\mathcal{E}(x^0, x^i) - \mathcal{E}(0, x^i)}  \leq \frac{C}{M}
\]
and the result follows on noting
\[ 
\norm{\mathcal{E} - \frac{1}{2} (\p_0 \rho)^2- \frac{1}{2} M^2 \rho^2}{L^\infty(\mathcal{M}_T)}  \leq \frac{C}{M}.\qedhere
\]
\end{proof}

Now we turn to the proof of Theorem \ref{Prop3}. Let $\eta \in C^\infty_c(\mathcal{M}_T)$ and define the auxiliary function $\chi$ again through \eq{chidef}. By \eq{WEcomp} we have:
\[
\int_{\mathcal{M}_T} \theta\eta = 2\int_{\mathcal{M}_T} \p_\mu \rho \p^\mu \theta\, \chi + \int_{\{x^0=0\}} \left[\ \theta_1 \chi - \theta_0 (\p_0 \chi)  \right]. 
\]
Integrating the first term on the right-hand-side by parts we find:
\begin{align*}
\int_{\mathcal{M}_T} \theta\eta &= \int_{\{x^0=0\}} \left[\ (\theta_1+ 2 \rho_0 \theta_1) \chi - \theta_0 (\p_0 \chi)  \right] - 2\int_{\mathcal{M}_T} \left(\rho \Box \theta\, \chi + \rho \p_\mu \theta \p^\mu \chi \right). \\
&= \int_{\{x^0=0\}} \left[\ (\theta_1+ 2 \rho_0 \theta_1) \chi - \theta_0 (\p_0 \chi)  \right] + \int_{\mathcal{M}_T} \left(4 \rho \p_\mu \rho \p^\mu \theta\, \chi - 2 \rho \p_\mu \theta \p^\mu \chi \right)
\end{align*}
here we have used the equation for $\theta$ to replace $\Box \theta$ on the right hand-side. We can continue in this vein:
\begin{align*}
\int_{\mathcal{M}_T} \theta\eta&= \int_{\{x^0=0\}} \left[\ \theta_1(1+ 2 \rho_0 + 2 \rho_0^2 ) \chi - \theta_0 (\p_0 \chi)  \right] - \int_{\mathcal{M}_T} \left(2\rho^2 \Box \theta\, \chi + 2 (\rho +\rho^2)\p_\mu \theta \p^\mu \chi \right)\\
&= \int_{\{x^0=0\}} \left[\ \theta_1\left(1+ 2 \rho_0 + 2 \rho_0^2 +\frac{4}{3} \rho_0^3\right) \chi - \theta_0 (\p_0 \chi)  \right] \\&\qquad - \int_{\mathcal{M}_T} \left( (2 \rho +2\rho^2 +\frac{4}{3} \rho^3 )\p_\mu \theta \p^\mu \chi - \frac{8}{3} \rho^3 \p_\mu \rho \p^\mu\theta \chi \right)\\
&=\int_{\mathcal{M}_T} \left( \tilde{\theta}^0 + \frac{1}{M}\tilde{\theta}^1+ \frac{1}{M^2}\tilde{\theta}^{2,0}   \right)\eta - 2\int_{\mathcal{M}_T} (\rho+\rho^2) \p_\mu \theta \p^\mu \chi  + O(M^{-3}),
\end{align*}
where 
\[
\begin{array}{crcrc}
\displaystyle- \Box \tilde{\theta}^0 = 0,&\qquad& \displaystyle- \Box \tilde{\theta}^1 = 0&\qquad & \displaystyle- \Box \tilde{\theta}^{2,0} = 0, \\[.4cm]
\tilde{\theta}^0 |_{x^0=0} = \theta_0, &\qquad & \tilde{\theta}^1 |_{x^0=0} = 0,&\qquad & \tilde{\theta}^{2,0} |_{x^0=0} = 0, \\[.4cm]
\p_0 \tilde{\theta}^0 |_{x^0=0} = \theta_1& \qquad &\p_0 \tilde{\theta}^1 |_{x^0=0} = 2M \rho_0 \theta_1& \qquad &\p_0 \tilde{\theta}^{2,0} |_{x^0=0} = 2 M^2 \rho_0^2  \theta_1.
\end{array}
\]
At this stage, we immediately have:
\ben{Conv}
\abs{\int_{\mathcal{M}_T} (\theta - \tilde{\theta}^0)\eta} \leq \frac{C_{data}}{M} \norm{\eta}{L^1_tL^2_x},
\een
by \eq{Eest}. From here it is immediate that $\theta = \tilde{\theta}^0 + O_{\mathscr{D}'}(M^{-1})$, thus to leading order in $M^{-1}$, the field $\theta$ satisfies the wave equation with uncorrected initial data. In order to proceed further in the expansion, we must deal with the bulk term above. To do this, we re-write the $\rho$ equation as:
\[
M^2(\rho + \rho^2) = \Box \rho + \p_\mu \rho \p^\mu \rho - \p_\mu \theta \p^\mu \theta - M^2 \rho^2 \frac{e^{2 \rho} -1-2\rho - 2 \rho^2}{2 \rho^2}
\]
so that:
\[
\int_{\mathcal{M}_T} (\rho+\rho^2) \p_\mu \theta \p^\mu \chi = \frac{1}{M^2} \int_{\mathcal{M}_T}\left( \Box \rho + \p_\nu \rho \p^\nu \rho - \p_\nu \theta \p^\nu \theta\right) \p_\mu \theta \p^\mu \chi + O(M^{-3})
\]
We take the terms one at a time, and integrate by parts as follows:
\begin{align*}
 \int_{\mathcal{M}_T}\Box \rho  \p_\mu \theta \p^\mu \chi &= \int_{\{x^0=0\}} \left[ \rho_1 (\p_i \theta_0 \p_i \chi - \theta_1 \p_0 \chi) - \rho \p_0(\p_\mu \theta \p^\mu \chi)\right] \\
&\quad +  \int_{\mathcal{M}_T}\rho \left( \p_\mu \Box  \theta \p^\mu \chi  +2 \p_\mu \p_\nu   \theta \p^\mu \p^\nu \chi + \p_\mu  \theta \p^\mu \Box  \chi  \right) \\
&= -\int_{\{x^0=0\}}  \left[ \p_i  (\rho_1 \p_i \theta_0)\chi + \rho_1 \theta_1 \p_0 \chi\right]  \\
&\quad -2  \int_{\mathcal{M}_T}\rho  \p_\mu (\p_\nu \rho \p^\nu \theta) \p^\mu \chi  +O(M^{-1})
\end{align*}
The term on the last line we split and again integrate by parts to obtain:
\begin{align*}
-2  \int_{\mathcal{M}_T}\rho  \p_\mu (\p_\nu \rho \p^\nu \theta) \p^\mu \chi &=   \int_{\mathcal{M}_T} \left(\frac{1}{2}\p_\mu \rho  \p^\mu \rho \p_\nu \theta \p^\nu \chi - \frac{3}{2}\rho \Box \rho \p_\mu \theta \p^\mu \chi \right)
\\ &\qquad  + \int_{\mathcal{M}_T}\left(  \frac{1}{2} \p_\sigma \rho \p_\tau \rho - \frac{3}{2} \rho \p_\sigma \p_\tau \rho \right) \p_\nu \theta \p_\mu \chi(g^{\sigma \mu}g^{\tau \nu} - g^{\sigma \tau} g^{\mu \nu})  +O(M^{-1})
\end{align*}
Now, we note that the term on the bottom line contains terms with at most one $\p_0$ derivative acting on $\rho$ multiplied by $\rho$ with at most spatial derivatives acting on it, so we see that it is in fact $O(M^{-1})$. Thus using the equation for $\rho$ to replace $\Box \rho$ we have
\begin{align*}
-2  \int_{\mathcal{M}_T}\rho  \p_\mu (\p_\nu \rho \p^\nu \theta) \p^\mu \chi &=   -\int_{\mathcal{M}_T} \left(\frac{1}{2}(\p_0 \rho)^2 + \frac{3}{2} M^2 \rho^2  \p^\mu \rho\right) \p_\nu \theta \p^\nu \chi  +O(M^{-1})
\end{align*}
Thus, we have:
\begin{align*}
\int_{\mathcal{M}_T}\left( \Box \rho + \p_\nu \rho \p^\nu \rho\right) \p_\mu \theta \p^\mu \chi &= -\int_{\{x^0=0\}}  \left[  \p_i  (\rho_1 \p_i \theta_0)\chi + \rho_1 \theta_1 \p_0 \chi\right] \\
&\qquad - \frac{3}{2}\int_{\mathcal{M}_T}\left((\p_0 \rho)^2  + M^2 \rho^2 \right) \p_\mu \theta \p^\mu \chi + O(M^{-1})
\end{align*}

Now, by Lemma \ref{EngLem} we can replace $(\p_0 \rho)^2  + M^2 \rho^2$ with $2\varepsilon$ to obtain:
\begin{align*}
\int_{\mathcal{M}_T}\left( \Box \rho + \p_\nu \rho \p^\nu \rho\right) \p_\mu \theta \p^\mu \chi &= -\int_{\{x^0=0\}}  \left[  \p_i  (\rho_1 \p_i \theta_0)\chi + \rho_1 \theta_1 \p_0 \chi\right] \\
&\qquad - 3\int_{\mathcal{M}_T}\varepsilon \p_\mu \theta \p^\mu \chi + O(M^{-1})
\end{align*}

At the moment, $\theta$ still appears on the right-hand side. In order to remove this, we note that the argument leading to \eq{Conv} did not require that $\eta$ vanishes near $x^0=0$, only near $x^0=T$. Accordingly, we have (since the Cauchy data for $\theta$ and $\tilde{\theta}^0$ agree):
\[
\int_{\mathcal{M}_T}\varepsilon \p_\mu (\theta-\tilde{\theta}^0) \p^\mu \chi = -\int_{\mathcal{M}_T}(\theta-\tilde{\theta}^0) \p_\mu(\varepsilon \p^\mu \chi) = O(M^{-1}).
\]
Thus
\begin{align*}
\int_{\mathcal{M}_T}\varepsilon \p_\mu \theta \p^\mu \chi&= \int_{\mathcal{M}_T}\varepsilon \p_\mu \tilde{\theta}^0 \p^\mu \chi +O(M^{-1}) \\
&= \int_{\{x^0=0\}} \varepsilon \theta_1  \chi -  \int_{\mathcal{M}_T}\p_\mu\left(\varepsilon \p_\mu \tilde{\theta}^0 \right) \chi +O(M^{-1})
\end{align*}
Finally then, we have established:
\begin{align*}
-2\int_{\mathcal{M}_T}\left( \Box \rho + \p_\nu \rho \p^\nu \rho\right) \p_\mu \theta \p^\mu \chi &= 2\int_{\{x^0=0\}}  \left[  \left(\p_i  (\rho_1 \p_i \theta_0)+3\varepsilon \theta_1\right)\chi + (\rho_1 \theta_1) \p_0 \chi\right] \\
& \qquad - 6\int_{\mathcal{M}_T}\p_\mu\left(\varepsilon \p_\mu \tilde{\theta}^0 \right) \chi +O(M^{-1}) \\
&= \int_{\mathcal{M}_T} \tilde{\theta}^{2,1} \eta +O(M^{-1}) 
\end{align*}
where:
\[
\begin{array}{c}
\displaystyle- \Box \tilde{\theta}^{2,1} = -6 \p_\mu\left(\varepsilon \p_\mu \tilde{\theta}^0 \right), \\[.4cm]
\tilde{\theta}^{2,1} |_{x^0=0} = -2\rho_1 \theta_1, \\[.4cm]
\p_0 \tilde{\theta}^{2,1} |_{x^0=0} =2\p_i  (\rho_1 \p_i \theta_0)+6 \varepsilon \theta_1.
\end{array}
\]

The final term we need to consider is:
\[
2\int_{\mathcal{M}_T} \p_\nu \theta \p^\nu \theta \p_\mu \theta \p^\mu \chi
\]
By a similar argument to that above, we have:
\begin{align*}
2\int_{\mathcal{M}_T} \p_\nu \theta \p^\nu \theta \p_\mu \theta \p^\mu \chi &= 2\int_{\mathcal{M}_T} \p_\nu \tilde{\theta}^0 \p^\nu \tilde{\theta}^0 \p_\mu \tilde{\theta}^0 \p^\mu \chi +O(M^{-1}) \\
&= -2\int_{\mathcal{M}_T}  \p^\mu\left( \p_\nu \tilde{\theta}^0 \p^\nu \tilde{\theta}^0 \p_\mu \tilde{\theta}^0 \right) \chi  -2 \int_{\{x^0=0\}} (\theta_1^2 - \p_i \theta_0 \p^i \theta_0) \theta_1\chi +O(M^{-1})\\
&=\int_{\mathcal{M}_T} \tilde{\theta}^{2,2} \eta +O(M^{-1}) 
\end{align*}
where:
\[
\begin{array}{c}
\displaystyle- \Box \tilde{\theta}^{2,2} = -2\p^\mu\left( \p_\nu \tilde{\theta}^0 \p^\nu \tilde{\theta}^0 \p_\mu \tilde{\theta}^0 \right), \\[.4cm]
\tilde{\theta}^{2,2} |_{x^0=0} = 0, \\[.4cm]
\p_0 \tilde{\theta}^{2,2} |_{x^0=0} = 2( \p_i \theta_0 \p^i \theta_0- \theta_1^2) \theta_1.
\end{array}
\]
Combining all of the results above, we have established Theorem \ref{Prop3}.\qedhere

\section{Numerical Verification}\label{Numerics}

\subsection{Testing Theorem \ref{Thm:EFTApr}}
\begin{figure}[ht]
\includegraphics[width=\textwidth]{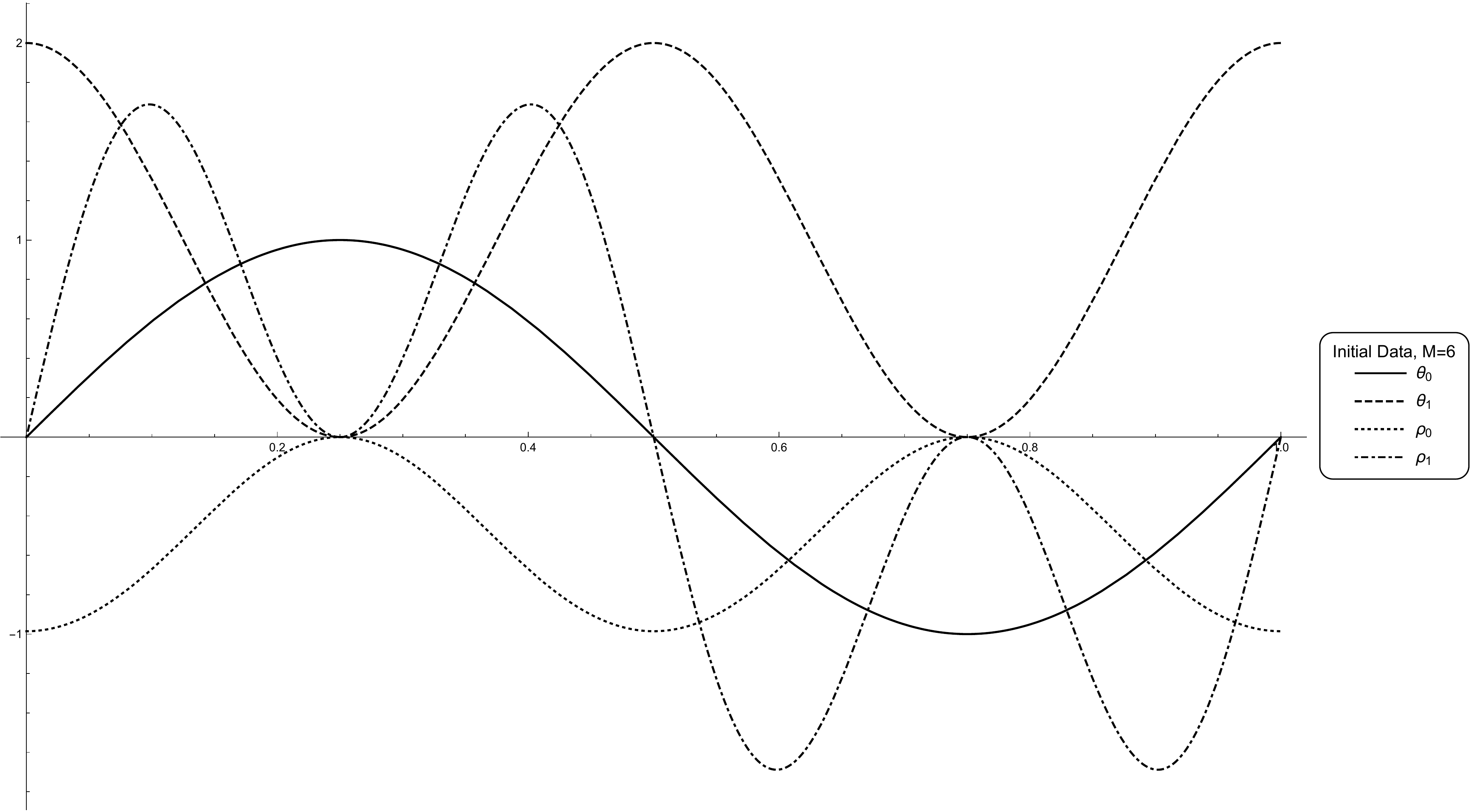}
\caption{Initial conditions chosen for numerical evolution to test Proposition \ref{Prop2}} \label{Fig1}
\end{figure}
We have performed some basic numerics to confirm our analytical results above. We investigated how well solutions to the equation 
\begin{align}
- \Box \tilde{\theta}  &=  \frac{2}{M^2} \p^\mu (- \p_\mu \tilde{\theta}  \cdot \p^\mu \tilde{\theta} ) \p_\mu \tilde{\theta}, \label{EFT2again}
\end{align}
approximate true solutions to the coupled $\rho, \theta$ equation.

Firstly, for a particular choice of initial data $\theta_0$, $\theta_1$ for the $\theta$ field, using Mathematica's built-in numerical routines, we solved  \eq{EFT2again} for $(x, t) \in [0,1]\times [0, 2]$ with initial conditions $\tilde{\theta}|_{\{x^0 = 0\}} = \theta_0, \p_0\tilde{\theta}|_{\{x^0 = 0\}} = \theta_1$ and periodic boundary conditions.

Then we solved \eq{rhothetaeq1}, \eq{rhothetaeq2} with periodic boundary conditions on the same domain, with ${\theta}|_{\{x^0 = 0\}} = \theta_0, \p_0{\theta}|_{\{x^0 = 0\}} = \theta_1$, and initial conditions:
\begin{enumerate}[1)]
\item ${\rho}|_{\{x^0 = 0\}} = \p_0{\rho}|_{\{x^0 = 0\}} = 0$ ,
\item ${\rho}|_{\{x^0 = 0\}} = M^{-2}(\theta_1^2 - \p_i \theta_0 \p_i \theta_0)$,  $\p_0{\rho}|_{\{x^0 = 0\}} = (60/M)^2(1 + \sin(2\pi x))$,
\item ${\rho}|_{\{x^0 = 0\}} = M^{-2}(\theta_1^2 - \p_i \theta_0 \p_i \theta_0)$,  $\p_0{\rho}|_{\{x^0 = 0\}} = 2 M^{-2}( \theta_1\p_x\p_x\theta_0 - \p_x \theta_0 \p_x \theta_1)$.
\end{enumerate}
These choices of initial condition for the $\rho$ field are such that we have a bound for $N_{k, 1}$, $N_{k, 2}$, $N_{k, 3}$ respectively, which is   uniform in $M$. The reason for the  slightly strange choice in $2)$ is that within the range of $M$ available to us numerically, it seems to be difficult to distinguish between $2)$ and $3)$ if we choose (for example) $\p_0{\rho}|_{\{x^0 = 0\}} = 0$. The particular choices of $\theta_0$ and $\theta_1$, as well as the initial data for $\rho$ in case $3)$ are plotted in Figure \ref{Fig1}.

To capture the error that we make by using \eq{EFT2again} in place of the full system, we computed for each choice of initial condition for $\rho$:
\[
e= \sqrt{ \int_0^1 dx \left |\theta - \tilde{\theta}\right|_{t = 0.7}^2}.
\]
In this way we estimate the  $C^0_tL^2_x$ norm by sampling at a particular time.

Figure \ref{Fig2} shows a plot of $e_1, e_2, e_3$, the error corresponding to choices $1, 2, 3$ above, against $M$ for a range of values with a log--log scale together with  fit lines corresponding to $e_l \sim M^{-l-1}$. We see that the numerics appear to be consistent with the results of  Theorem \ref{Thm:EFTApr}. This emphasises that the expansion for the field $\rho$ in terms of $\theta$ must hold on the initial data if one wishes the EFT solution to approximate the true solutions at higher orders.
\begin{figure}
\includegraphics[width=\textwidth]{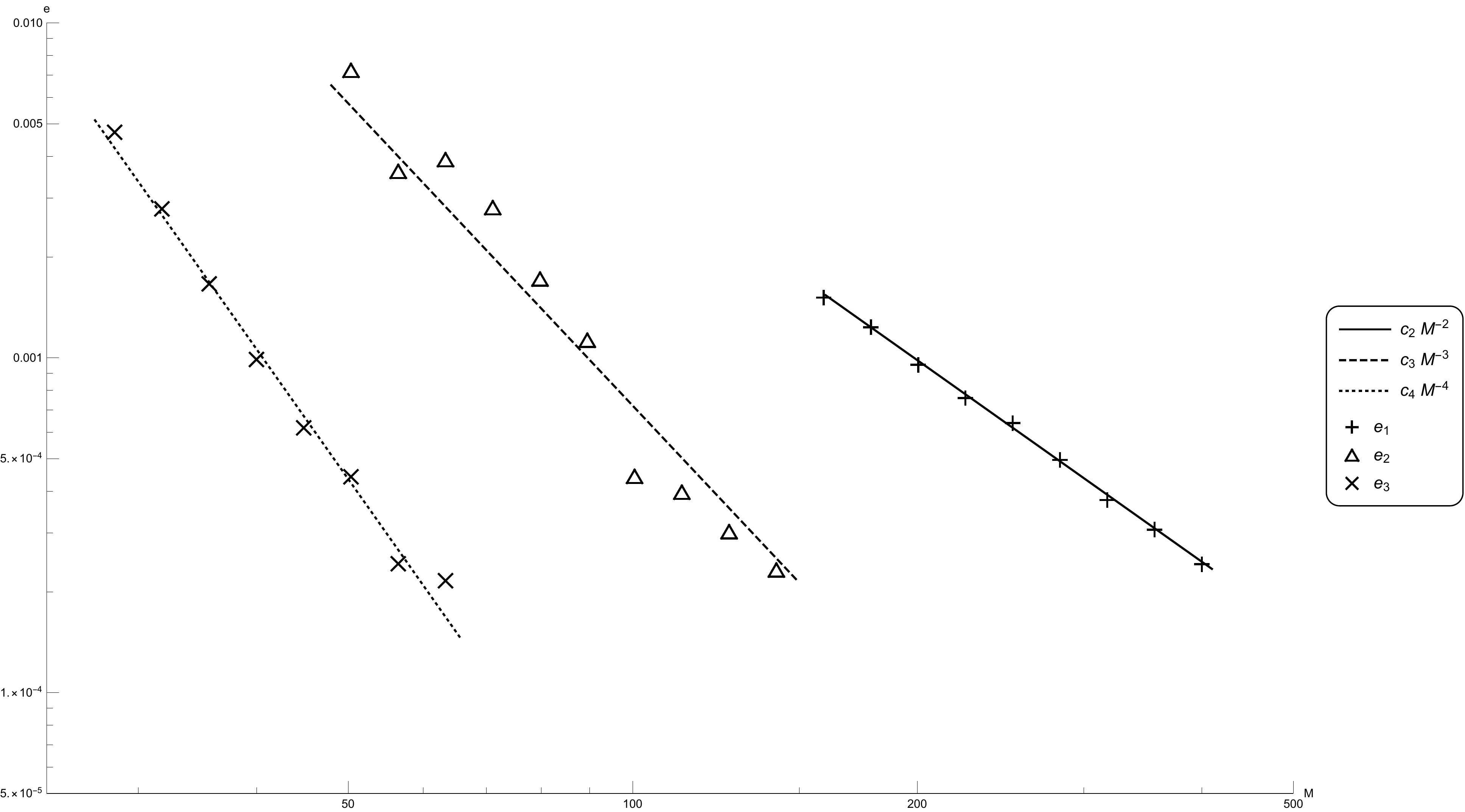}
\caption{Numerical computation of the errors (y-axis) for different values of $M$ (x-axis) for the EFT approach}\label{Fig2}
\end{figure}

\subsection{Testing Theorem \ref{Prop3}}

Using Mathematica's built-in numerical routines, we have solved  the coupled system of equations \eq{rhothetaeq1}, \eq{rhothetaeq2} for $(x, t) \in [0,1]\times [0, 2]$ with periodic boundary conditions and initial conditions shown in Figure \ref{Fig3}. We then computed the error terms:
\[
\begin{array}{l}
e_0 = \abs{\int_{t_1}^{t_2}dt \int_{x_1}^{x_2} dx (\theta - \tilde{\theta}^0) \eta}\\[.2cm]
e_1 = \abs{\int_{t_1}^{t_2}dt \int_{x_1}^{x_2} dx (\theta - \tilde{\theta}^0 -M^{-1}\tilde{\theta}^1 ) \eta }\\[.2cm]
e_2 = \abs{\int_{t_1}^{t_2}dt \int_{x_1}^{x_2} dx (\theta - \tilde{\theta}^0 -M^{-1}\tilde{\theta}^1-M^{-2}\tilde{\theta}^2 ) \eta} \\[.2cm]
\hat{e}_2 = \abs{\int_{t_1}^{t_2}dt \int_{x_1}^{x_2} dx (\theta - \hat{\theta}^2) \eta} \\
\hat{e}^{EFT}_2 = \abs{\int_{t_1}^{t_2}dt \int_{x_1}^{x_2} dx (\theta - \tilde{\theta}) \eta}
\end{array}
\]
Here $\tilde{\theta}^0, \tilde{\theta}^1, \tilde{\theta}^2$ were computed by numerically solving \eq{LinApr1}, \eq{LinApr2} with the appropriate initial data and $\hat{\theta}^2$ by solving \eq{NLinApr}. $\theta^{EFT}$ was computed by numerically solving the naive EFT equation to second order \eq{EFT2again} with the same initial conditions.

The weight function $\eta$ was given by:
\[
\eta(x, t) = n (x-x_1)(x_2-x)(t-t_1)(t_2-t)
\]
with the normalisation factor, $n$, chosen to make $\int \eta =1$. In view of the fact we see at most one derivative of $\eta$ appearing in the errors in \S\ref{AvgProof}, this choice of $\eta$ is sufficiently regular for our purposes. We chose $x_1 = 0.3, x_2 = 0.4, t_1 = 1.2, t_2 = 1.3$, so that our average is taken over a $0.1 \times 0.1$ window in space-time. Our results appear to be robust to changes in initial data or the size and location of the averaging window, however we do not claim to have undertaken a fully systematic numerical exploration of the system.

Figure \ref{Fig4} shows the error terms for various values of $M$, plotted on a log--log scale, together with fit lines corresponding to powers of $M$. The numerics appear consistent with our expectation that $e_i \sim M^{-i-1}$ and $\hat{e}_2 \sim M^{-3}$. We also see clearly that the naive second-order corrected EFT \eq{EFT2again} offers no real improvement over the uncorrected EFT (equivalent to $\tilde{\theta}^0$) in this regime where the field $\rho$ is not assumed to satisfy the EFT expansion at the level of initial data. This emphasises once again our results above that the EFT expansion is only valid when initial data are chosen appropriately.

We observe that the nonlinear approximant $\hat{\theta}^2$ appears to give a better approximation than the iterated linear approximant $\tilde{\theta}^0 +M^{-1}\tilde{\theta}^1+M^{-2}\tilde{\theta}^2$. Given that we have not considered a wide range of initial data, this may simply be an artefact of the choice of initial conditions. One could also imagine that this is a result of compounding numerical errors for the iterated approximant by solving three linear equations rather than one nonlinear equation, however experimentation with the numerical routines (modifying step-size, accuracy goals) does not appear to support this conclusion. We note also that the second order approximations appear noisier than the lower order approximations. This is likely due to the fact that the improvement in accuracy at this level comes about through cancellation in the averaging integral, which can be difficult to capture numerically. 

Our numerical testing of Theorem \ref{Thm:EFTApr} and Theorem \ref{Prop3} give us confidence that the analysis above has been carried out correctly, and moreover that the powers of $M$ obtained in our estimates are sharp.

\begin{figure}[t]
\includegraphics[width=\textwidth]{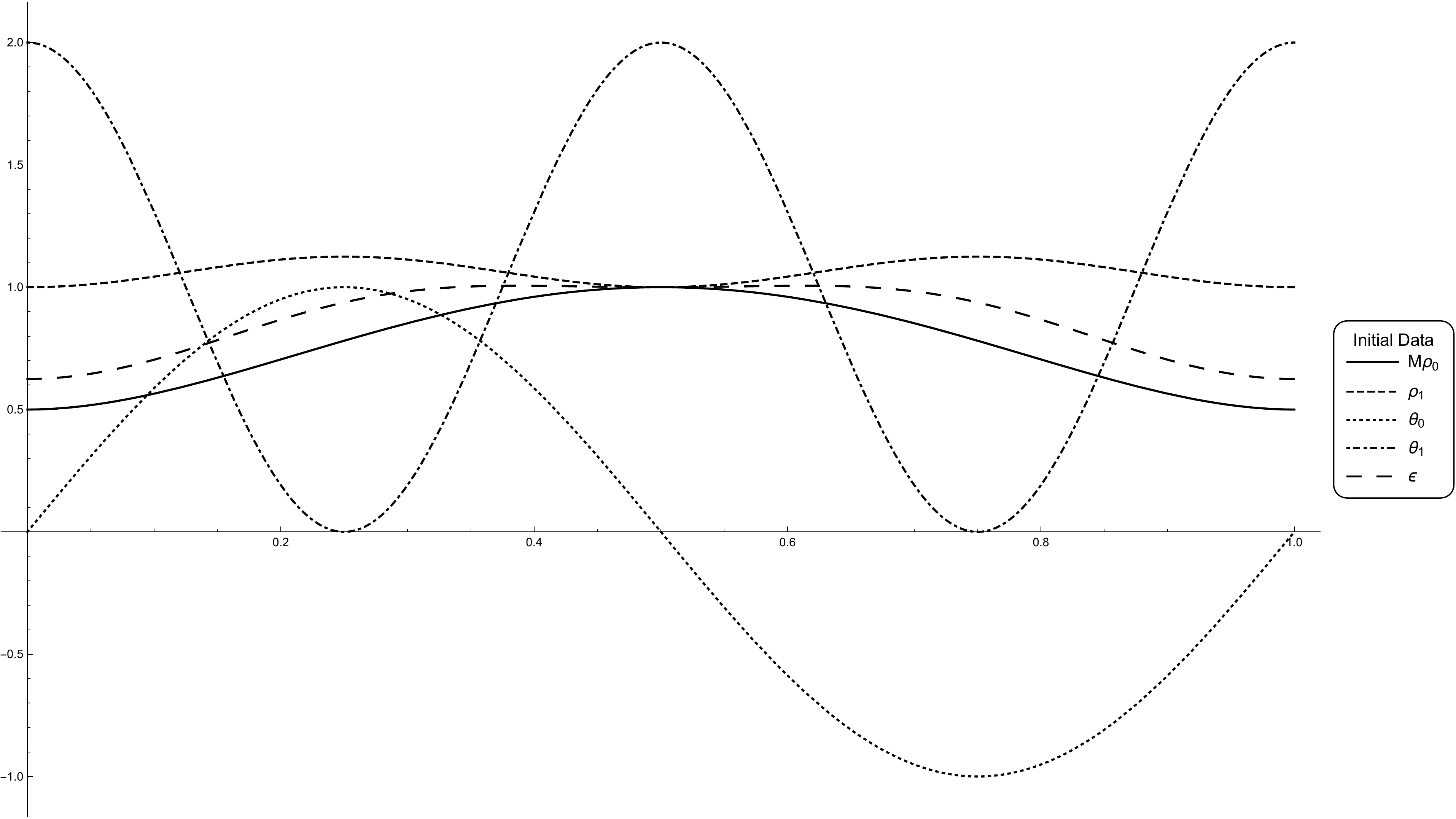}
\caption{Initial conditions chosen for numerical evolution to test Theorem \ref{Prop3}} \label{Fig3}
\end{figure}
\begin{figure}
\includegraphics[width=\textwidth]{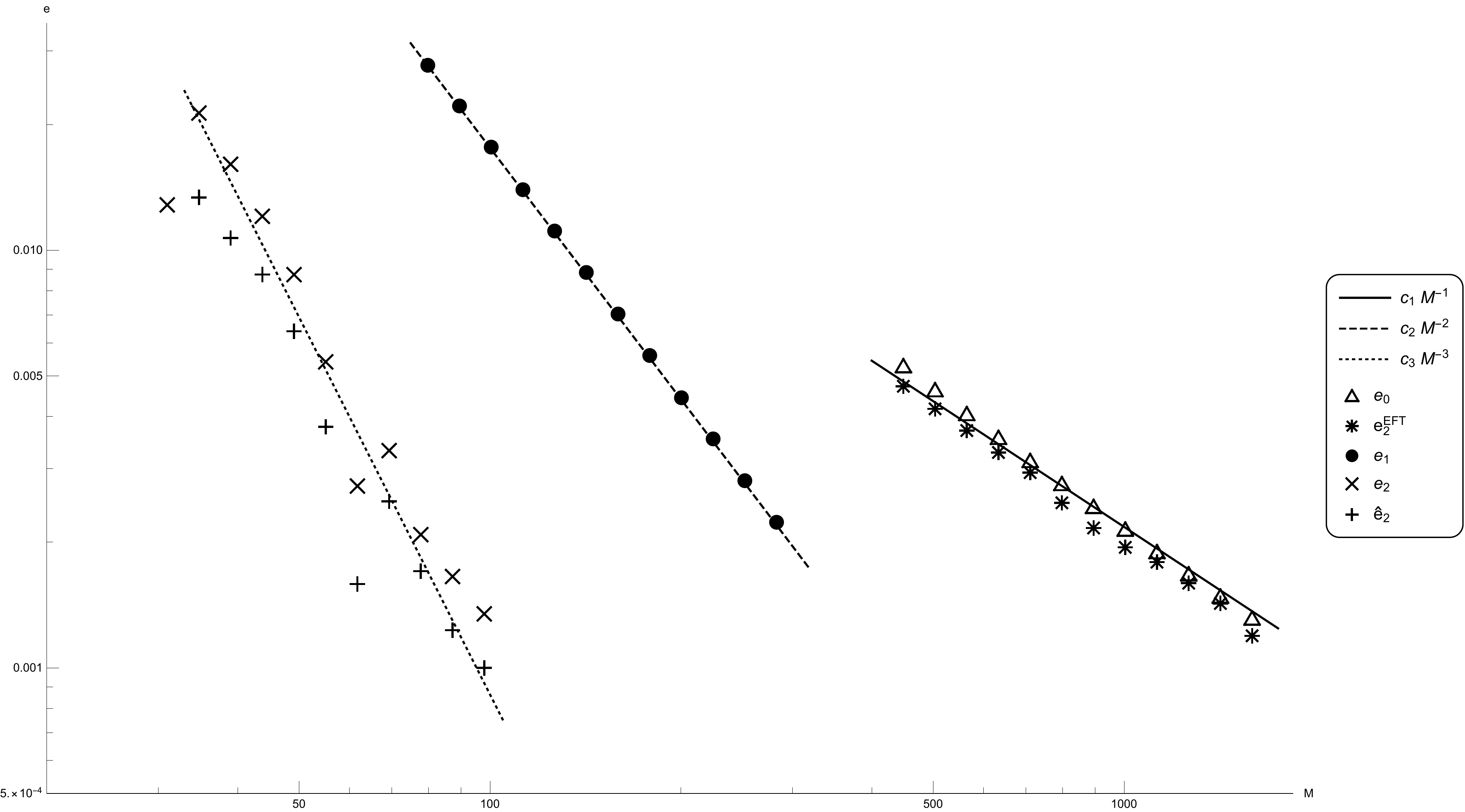}
\caption{Numerical computation of the errors (y-axis) for different values of $M$ (x-axis) for the averaging approach}\label{Fig4}
\end{figure}

\bibliographystyle{utphys}
\bibliography{EFTBib}
\end{document}